\pdfoutput=1
\documentclass[12pt,a4paper]{article}
\usepackage{amsmath,amssymb,amsthm}
\usepackage[english]{babel}
\usepackage[T1]{fontenc}
\usepackage[utf8]{inputenc}
\usepackage{csquotes}
\usepackage{url}
\usepackage[numbers, square]{natbib}
\usepackage{mathrsfs}
\usepackage{tikz}
\usetikzlibrary{arrows}
\usetikzlibrary{shapes.geometric,calc,positioning}
\usepackage{mathtools}
\usepackage{ stmaryrd }
\usepackage{subcaption}
\usepackage{pseudocode}
\allowdisplaybreaks[1]
\usepackage{hyperref}
\usepackage[nottoc,numbib]{tocbibind}
\usepackage{geometry}	
\usepackage{graphicx}

\newcommand\defeq{\stackrel{\mathclap{\normalfont\mbox{\tiny def}}}{=}}

\usepackage{eso-pic}
\newcommand\BackgroundPic{%
\put(0,0){%
\parbox[b][\paperheight]{\paperwidth}{%
\vfill
\centering
\includegraphics[width=\paperwidth,height=\paperheight,%
keepaspectratio]{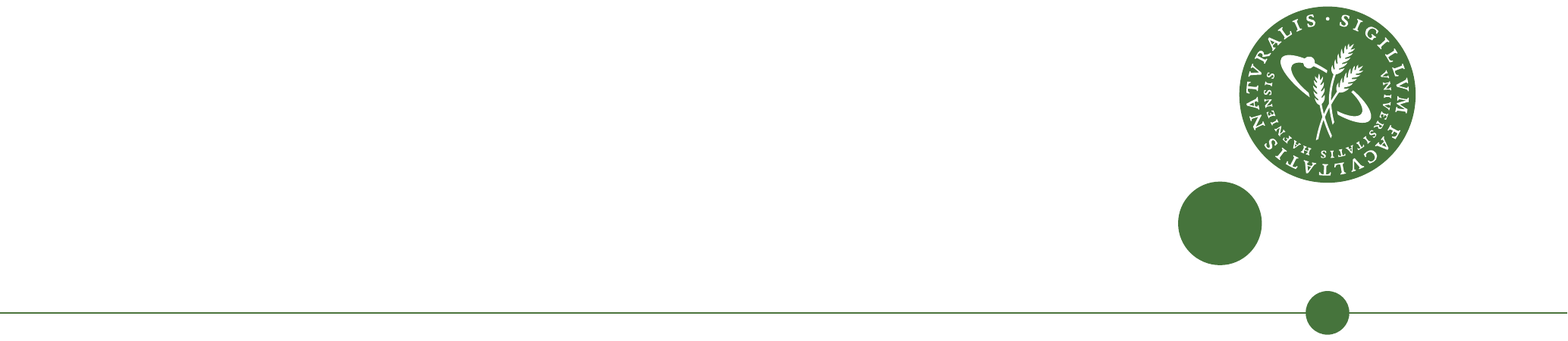}%
\vfill
}}}

\usepackage{fourier-orns}								% beautiful ornaments
\usepackage{soul}										% custom spacing in/between words
\sodef\an{}{0.1em}{0.5em}{0em}							% custom spacing in/between words in titles
\sodef\ann{}{0.1em}{0.5em}{0em}							% custom spacing in/between words in smaller text

\theoremstyle{plain}
\newtheorem{thm}{Theorem}
\newtheorem{lem}[thm]{Lemma}
\newtheorem{prop}[thm]{Proposition}
\newtheorem{cor}[thm]{Corollary}
\newtheorem{defn}[thm]{Definition}

\theoremstyle{definition}

\theoremstyle{remark}
\newtheorem*{rem}{Remark}
\newtheorem*{note}{Note}

\newcommand\norm[1]{\left\lVert#1\right\rVert}

\title{\vspace{-2.0cm} How to Generate Pseudorandom Permutations Over Other Groups} 
\author{Hector B. Hougaard}
\date{\vspace{-5ex}}

\begin{document}

\begin{titlepage}
\AddToShipoutPicture*{\BackgroundPic} % University/department logo
\thispagestyle{empty}

\newgeometry{margin=25mm,bottom=30mm}
\fontfamily{ptm}\selectfont
\scshape

\vspace{65pt}

{\centering	% vertically center each block (separate with \vfill)
	
	\vspace*{130pt}
	
	\newcommand{\ts}{31}	% size of title text
	
	\fontsize{14}{14}\selectfont
	\an{How to} \par
	\fontsize{\ts}{\ts}\selectfont
	\ann{Generate Pseudorandom Permutations} \par \vspace{4pt}
	\fontsize{\ts}{\ts}\selectfont
	\an{Over Other Groups}
	%	\an{Graph }$C^*$-\an{algebras}
	
	\vspace{26pt}
	
	\fontsize{16}{16}\selectfont
	\ann{Hector Bjoljahn Hougaard}\par \vspace{.6em}
	\fontsize{12}{12}\selectfont
	\ann{August 2017}

	\vfill
	
	\fontsize{14}{14}\selectfont
	\ann{Advisors:} \\
	\ann{Gorjan Alagic} \\
	\fontsize{10}{10}\selectfont
	\ann{and} \\
	\fontsize{14}{14}\selectfont
	\ann{Florian Speelman}

	\vspace{42pt}
	
	\fontsize{10}{16}\selectfont
	\ann{Master}'s\ann{ Thesis in Mathematics} \\
	\ann{Department of Mathematical Sciences, University of Copenhagen}
	
}

\normalfont
\restoregeometry
\cleardoublepage
\end{titlepage}

\newpage
\pagenumbering{gobble}
\begin{abstract}
Recent results by Alagic and Russell have given some evidence that the Even-Mansour cipher may be secure against quantum adversaries with quantum queries, if considered over other groups than $(\mathbb{Z}/2)^n$. This prompts the question as to whether or not other classical schemes may be generalized to arbitrary groups and whether classical results still apply to those generalized schemes.

In this thesis, we generalize the Even-Mansour cipher and the Feistel cipher. We show that Even and Mansour's original notions of secrecy are obtained on a one-key, group variant of the Even-Mansour cipher. We generalize the result by Kilian and Rogaway, that the Even-Mansour cipher is pseudorandom, to super pseudorandomness, also in the one-key, group case. Using a Slide Attack we match the bound found above. After generalizing the Feistel cipher to arbitrary groups we resolve an open problem of Patel, Ramzan, and Sundaram by showing that the $3$-round Feistel cipher over an arbitrary group is not super pseudorandom.

We generalize a result by Gentry and Ramzan showing that the Even-Mansour cipher can be implemented using the Feistel cipher as the public permutation. In this result, we also consider the one-key case over a group and generalize their bound.

Finally, we consider Zhandry's result on quantum pseudorandom permutations, showing that his result may be generalized to hold for arbitrary groups. In this regard, we consider whether certain card shuffles may be generalized as well.
\end{abstract}

\newpage
\pagenumbering{gobble}
\begin{center}
\vspace*{130pt}
\textit{Tak til min mor, \\}
\textit{min far, \\}
\textit{Finn og Dorte. \\}
\vspace*{10pt}
\textit{Med hele mit hjerte,\\}
\textit{tak for alt Jeres støtte.}
\end{center}

\newpage
\pagenumbering{gobble}
\tableofcontents

\newpage
\pagenumbering{arabic}
\section{Introduction}
In many branches of mathematics, physics, economics, and others, randomness is of great importance. In mathematics, it is generally assumed, especially in probability theory, that true randomness exists. The question therefore becomes whether or not such true randomness exists in reality. In cryptography as well, it is assumed that true randomness exists, lest all cryptographic security falls apart. However, the authenticity of sources of true randomness can be debated. Cryptography offers an alternative to true randomness: pseudorandomness. Pseudorandomness is something which looks truly random to certain observers, but is not actually so. Over time, cryptographers have constructed many candidates for pseudorandom generators, pseudorandom functions, and pseudorandom permutations. One type of construction, is the so-called block cipher.

In \citep{EM}, Even and Mansour introduced and proved security for the DES inspired block cipher scheme we now call the Even-Mansour (EM) scheme. Given a public permutation, $P$, over $n$-bit strings, with two different, random, secret, $n$-bit keys $k_1$ and $k_2$, a message $x\in \lbrace 0,1 \rbrace^n$ could be enciphered as
\begin{align*}
EM_{k_1,k_2}^P (x) = P(x \oplus k_1)\oplus k_2,
\end{align*}
with an obvious decryption using the inverse public permutation. The scheme was minimal, in the sense that they needed to XOR a key before and after the permutation, otherwise the remaining key could easily be found. As an improvement, Dunkelman, Keller, and Shamir \citep{DKS} showed that there was only a need for a single key and the scheme would still retain an indistinguishability from random, i.e. it was pseudorandom. As another consideration of block ciphers, the Feistel cipher construction of Luby and Rackoff \citep{LubyR} showed how to build pseudorandom permutations from pseudorandom functions. 

Eventually, Kuwakado and Morii showed that both the EM scheme \citep{BrokenEM} and the Feistel scheme \citep{BrokenFeistel} could be broken by quantum adversaries with quantum queries. Rather than discard these beautiful constructions entirely, Alagic and Russell \citep{Gorjan} considered whether it would be possible to define the two-key EM scheme over abelian groups in order to retain security against quantum adversaries with quantum queries. What they showed was a security reduction to the Hidden Shift Problem, over certain groups, such as $\mathbb{Z}/2^n$ and $S_n$. This result inspires us to ask whether the EM and Feistel schemes can be generalized over all groups, and if so, whether or not we can get pseudorandomness in some model.

\subsection{Prior Work}
In extension of their simplification of the EM scheme, Dunkelman, Keller, and Shamir \citep{DKS} attacked the construction using variants of slide attacks in order to show that the security bound was optimal. They further considered other variants of the EM scheme, such as the Addition Even-Mansour with an Involution as the Permutation (two-keyed). Also Kilian and Rogaway \citep{Kilr} were inspired by DESX and EM to define their $FX$ construction, of which the EM scheme is a special case.

As referred to above, Kuwakado and Morii were able to break the EM scheme \citep{BrokenEM} and the $3$-round Feistel scheme \citep{BrokenFeistel} on $n$-bit strings, using Simon's algorithm, if able to query their oracle with a superposition of states. Kaplan et al. \citep{Kaplan}, using Kuwakado and Morii's results, showed how to break many classical cipher schemes, which in turn incited Alagic and Russell \citep{Gorjan}.

In their work with the Hidden Shift Problem, \citep{Gorjan} posit that a Feistel cipher construction over other groups than the bit strings group might be secure against quantum adversaries with quantum queries. Many Feistel cipher variants exist, with different relaxations on the round functions, see for example \citep{NaorReingold} and \citep{PatelRamzanSundaram}, the latter of which also considered Feistel ciphers over other groups. Vaudenay \citep{Vaudenay} also considered Feistel ciphers over other groups in order to protect such ciphers against differential analysis attacks by what he called decorrelation.

Removed from the schemes considered below and with a greater degree of abstraction, Black and Rogaway \citep{BlackRogaway} consider ciphers over arbitrary domains which Hoang, Morris, and Rogaway \citep{HMR, MR14} do as well. More generally, Zhandry \citep{Zhandry} defines the notion of Function to Permutation Converters in order to show that the existence of quantum one-way functions imply the existence of quantum pseudorandom permutations.

\subsection{Summary of Results}
We work in the Random Oracle Model and consider groups $G$ in the family of finite groups, $\mathcal{G}$. We consider pseudorandom permutations, given informally as the following.

\begin{defn}
\textbf{[Informal]} A keyed permutation $P$ on a group $G$ is a \textbf{Pseudorandom Permutation (PRP)} on $G$ if it is indistinguishable from a random permutation for all probabilistic distinguishers having access to only polynomially many permutation-oracle queries.
\end{defn}

A \textbf{Super Pseudorandom Permutation (SPRP)} is a permutation where the distinguisher is given access to the inverse permutation-oracle as well.

We define the \textbf{Group Even-Mansour (EM) scheme} on $G$ to be the encryption scheme having the encryption algorithm
\begin{align*}
E_k(m) = P(m\cdot k) \cdot k,
\end{align*}
where $m\in G$ is the plaintext and $k\in G$ is the uniformly random key.

We define two problems for the Group Even-Mansour scheme: \textbf{Existential Forgery} (EFP) and \textbf{Cracking} (CP). In EFP, the adversary must eventually output a plaintext-ciphertext pair which satisfies correctness. In CP, the adversary is given a ciphertext and asked to find the corresponding plaintext.

It holds that for our Group EM scheme, the probability that an adversary succeeds in the EFP is polynomially bounded:
\begin{thm}\label{IntuitiveEFP}
\textbf{[Informal]} If $P$ is a uniformly random permutation on $G$ and $k\in G$ is chosen uniformly at random. Then, for any probabilistic adversary $\mathcal{A}$, the success probability of solving the EFP is negligible, specifically, bounded by
\begin{align*}
O\left( \frac{st}{|G|} \right),
\end{align*}
where $s$ and $t$ are the amount of encryption/decryption- and permutation/inverse permutation-oracle queries, respectively.
\end{thm}

By a basic reduction, and for the latter, by an inference result, we also get that
\begin{thm}
\textbf{[Informal]} If $P$ is a super pseudorandom permutation on $G$ and $k\in G$ is chosen uniformly at random. For any probabilistic adversary $\mathcal{A}$, the success probability of solving the EFP is negligible.
\end{thm}

\begin{cor}
\textbf{[Informal]} If $P$ is a super pseudorandom permutation on $G$ and $k\in G$ is chosen uniformly at random. For any \underline{polynomial-time} probabilistic adversary $\mathcal{A}$, the success probability of solving the CP is negligible.
\end{cor}

With the same bound as in Theorem~\ref{IntuitiveEFP}, we find that

\begin{thm}
\textbf{[Informal]} For any probabilistic adversary $\mathcal{A}$, limited to polynomially many encryption- and decryption-oracle queries and polynomially many permutation- and inverse permutation-oracle queries, the Group EM scheme over a group $G$ is a super pseudorandom permutation.
\end{thm}

We then apply a Slide Attack, to find an attack which matches the bound given above.

Considering the \textbf{Group Feistel cipher}, whose encryption algorithm consists of multiple uses of the round function
\begin{align*}
\mathcal{F}_f(x,y) = (y,x\cdot f(y)),
\end{align*}
where $f$ is a pseudorandom function on $G$, we show that the \textbf{$3$-round Feistel cipher} is pseudorandom but is not super pseudorandom, regardless of the underlying group $G$. We then note that the $4$-round Feistel cipher is super pseudorandom as proven in \citep{PatelRamzanSundaram}.

Finally, we consider the \textbf{Group Even-Mansour scheme instantiated using a $4$-round Feistel cipher} over $G^2 = G\times G$, which uses the encryption algorithm
\begin{align*}
\Psi_k^{f,g}(m) = \mathcal{F}_{g,f,f,g}(m \cdot k)\cdot k,
\end{align*}
where $f$ and $g$ are modelled as random functions, $m\in G^2$ the plaintext, and $k \in G^2$ is a uniformly random key. We then show one of our main results:

\begin{thm}
\textbf{[Informal]} For any probabilistic $4$-oracle adversary $\mathcal{A}$ with at most
\begin{itemize}
\item $q_c$ queries to the $\Psi$- and inverse $\Psi$-oracles (or random oracles),
\item $q_f$ queries to the $f$-oracle, and
\item $q_g$ queries to the $g$-oracle,
\end{itemize}
we have that the success probability of $\mathcal{A}$ distinguishing between $\Psi$ and a random oracle, is bounded by
\begin{align*}
 (2q_c^2 +4q_fq_c + 4q_gq_c + q_c^2 - q_c)|G|^{-1} + 2\cdot \begin{pmatrix}
	q_c \\ 2
	\end{pmatrix}(2|G|^{-1} + |G|^{-2}).
\end{align*}
\end{thm}

We may also rewrite our main theorem as the following:
\begin{thm}
\textbf{[Informal]} For any $4$-oracle adversary $\mathcal{A}$, with at most $q$ total queries, we have that the success probability of $\mathcal{A}$ distinguishing between $\Psi$ and a random oracle, is bounded by
\begin{align*}
2(3q^2-q)|G|^{-1} + (q^2-q)|G|^{-2}.
\end{align*}
\end{thm}

We note that this main result is due to \citep{GentryRamzan}, however, we consider a one-key group version and add details to their proof sketches.

Using Zhandry \citep{Zhandry}, we generalize the notion of a \textbf{Function to Permutation Converter}:

\begin{defn}
\textbf{[Informal]} A Function to Permutation Converter (FPC) is a pair of efficient classical oracle algorithms $(F,R)$ where $F$ and $R$ are each others inverses, making at most polynomially many queries to an inner oracle, such that they are indistiguishable from a random permutation and its inverse.
\end{defn}

The FPCs are then used to show that the existence of quantum one-way functions implies the existence of quantum pseudorandom permutations. In order to show this, we make use of articles \citep{HMR, MR14} to show that certain FPCs exist for groups, by ordering the group elements. For the sake of generality, we show that the notion of a \textbf{Swap-or-Not shuffle} on a deck of cards of \citep{HMR} may be generalized to hold for arbitrary groups by defining a \textbf{Scoot-or-Not (SC) shuffle}. For the SC shuffle, we show that the distinguishing probability of an adversary making classical $q$ queries to an $r$ round implementation of the shuffle on $G$, is bounded by
\begin{align*}
\frac{8|G|^{3/2}}{r+4}\left( \frac{q+|G|}{2|G|}\right)^{r/4+1}.
\end{align*}

However, we are not able to generalize the Sometimes-Recurse shuffle of \citep{MR14} per group theoretic considerations.

\newpage
\section{Assumptions and General Definitions}\label{GenDefns}
Recall the following definition of a group, taken from \citep{DumbFoot}.

\begin{defn}\label{GroupDefn}
A \textbf{group} is an ordered pair $(G,\cdot)$ where $G$ is a set and $\cdot$ is a binary operation on $G$ satisfying the following axioms:
\begin{enumerate}
\item $(a\cdot b)\cdot c = a \cdot (b \cdot c)$, for all $a,b,c\in G$, i.e. $\cdot$ is \textbf{associative}.
\item there exists a unique \textbf{identity} element $1\in G$, such that for all $a\in G$, we have $a\cdot 1 = 1\cdot a = a$.
\item for each $a\in G$ there exists a unique \textbf{inverse} element $a^{-1}\in G$ such that $a\cdot a^{-1} = a^{-1} \cdot a = 1$.
\end{enumerate}
The group is called \textbf{abelian} (or \textbf{commutative}) if $a \cdot b = b\cdot a$ for all $a,b\in G$.
\end{defn}

In the following, we work in the Random Oracle Model such that we may assume the existence of a random permutation oracle on group elements. We let $\mathcal{G}$ be the family of all finite groups, e.g. a group $G\in\mathcal{G}$ is a pair of the set $G$ and operation $\cdot$ satisfying the group axioms. We also assume that for any group $G\in \mathcal{G}$, $|G|\leq 2^{poly(n)}$ for some $n\in \mathbb{N}$ and some polynomial $poly(\cdot)$.

We will need the concept of pseudorandom, which is also called indistinguishable from random, in several forms. On notation, we write $x\in_R X$ for an element chosen uniformly at random from a set $X$. In the following, we consider the positive integer $\lambda$ to be the security parameter, specified in unary per convention. We assume that for each $\lambda$ there exists a uniquely specified group $G(\lambda) = G_\lambda \in \mathcal{G}$ with size $|G_\lambda| \geq 2^\lambda$.

\begin{defn}
For an adversary $\mathcal{A}$, we write
\begin{align*}
\underset{x \in_R X}{Pr}\left[ \mathcal{A}^{O_x}(\lambda)=1 \right]
\end{align*}
to denote the probability that the adversary outputs $1$ as a function of $\lambda$ in a standard distinguishability experiment, when given access to the oracle $O_x$ which relies on the randomness of $x$.
\end{defn}

\begin{defn}
Let $F_{m,n}:G_\lambda \times G_m \rightarrow G_n$, for $G_m,G_n \in \mathcal{G}$, be an efficient, keyed function. $F_{m,n}$ is a \textbf{pseudorandom function (PRF)} if for all probabilistic distinguishers $\mathcal{A}$, limited to only polynomially many queries to the function-oracle, there exists a negligible function $negl(\cdot)$, such that
\begin{align*}
\left| \underset{k \in_R G_\lambda}{Pr}\left[ \mathcal{A}^{F_{m,n}(k,\cdot)}(\lambda)=1 \right] - \underset{\pi \in_R \mathfrak{F}_{G_m\rightarrow G_n}}{Pr}\left[ \mathcal{A}^{\pi(\cdot)}(\lambda)=1 \right] \right| \leq negl(\lambda),
\end{align*}
where $\mathfrak{F}_{G_m\rightarrow G_n}$ is the set of functions from $G_m$ to $G_n$.
\end{defn}

If $F:G \times G \rightarrow G$ is a pseudorandom function, we say that it is a \textbf{pseudorandom function on $G$}.

\begin{defn}
Let $P:G_\lambda \times G \rightarrow G$ be an efficient, keyed permutation. $P$ is a \textbf{pseudorandom permutation (PRP)} if for all probabilistic distinguishers $\mathcal{A}$, limited to only polynomially many queries to the permutation-oracle, there exists a negligible function $negl(\cdot)$, such that
\begin{align*}
\left| \underset{k \in_R G_\lambda}{Pr}\left[ \mathcal{A}^{P(k,\cdot)}(\lambda)=1 \right] - \underset{\pi \in_R \mathfrak{P}_{G\rightarrow G}}{Pr}\left[ \mathcal{A}^{\pi(\cdot)}(\lambda)=1 \right] \right| \leq negl(\lambda),
\end{align*}
where $\mathfrak{P}_{G\rightarrow G}$ is the set of permutations on $G$.
\end{defn}

\begin{defn}
Let $P:G_\lambda \times G \rightarrow G$ be an efficient, keyed permutation. $P$ is said to be a \textbf{super pseudorandom permutation (SPRP)} if for all probabilistic distinguishers $\mathcal{A}$, limited to only polynomially many queries to the permutation- and inverse permutation-oracles, there exists a negligible function $negl(\cdot)$, such that
\begin{align*}
\left| \underset{k \in_R G_\lambda}{Pr}\left[ \mathcal{A}^{P(k,\cdot),P^{-1}(k,\cdot)}(\lambda)=1 \right] - \underset{\pi \in_R \mathfrak{P}_{G\rightarrow G}}{Pr}\left[ \mathcal{A}^{\pi(\cdot),\pi^{-1}(\cdot)}(\lambda)=1 \right] \right| \leq negl(\lambda),
\end{align*}
where $\mathfrak{P}_{G\rightarrow G}$ is the set of permutations on $G$.
\end{defn}

A (super) pseudorandom permutation $P:G \times G \rightarrow G$ is said to be a \textbf{(super) pseudorandom permutation on $G$}.

\newpage
\section{Even-Mansour}\label{Even-MansourSection}
We first remark that the results in this section were initially proven in a project prior to the start of the thesis but were further worked on to complement this thesis. Thus we have chosen to include parts of it, while this inclusion accounts for the brevity in certain results. We begin by defining the one-key Even-Mansour scheme over arbitrary groups, which we will refer to as the Group EM scheme.

\begin{defn}
We define the \textbf{Group Even-Mansour scheme} to be the triple of a key generation algorithm, encryption algorithm, and decryption algorithm. The key generation algorithm takes as input the security parameter $1^\lambda$, fixes and outputs a group $G\in_R \mathcal{G}$ with $|G|\geq 2^\lambda$, and outputs a key $k\in_R G$. The encryption algorithm $E_k(m)$ takes as input the key $k$ and a plaintext $m\in G$ and outputs
\begin{align*}
E_k(m) = P(m\cdot k) \cdot k,
\end{align*}
where $P$ is the public permutation. The decryption algorithm $D_k(c)$ takes as input the key $k$ and a ciphertext $c\in G$ and outputs
\begin{align*}
D_k(c) = P^{-1}(c \cdot k^{-1}) \cdot k^{-1},
\end{align*}
where $P^{-1}$ is the inverse public permutation.
\end{defn}
This definition satisfies \textbf{correctness} as:
\begin{align*}
D_k(E_k(m)) &= P^{-1}(P(m \cdot k) \cdot k \cdot k^{-1} ) \cdot k^{-1} \\
			&= P^{-1}(P(m \cdot k)) \cdot k^{-1} \\
			&= m \cdot k \cdot k^{-1} \\
			&= m,
\end{align*}
and even satisfies $E_k(D_k(c))=c$, i.e. it is a permutation.

\subsection{Two Forms of Security for the Group EM Scheme}
In this subsection, we prove classical results about our new scheme. We do so by considering Even and Mansour's two notions of security: the Existential Forgery Problem and the Cracking Problem, the Cracking Problem being the stronger of the two.

\begin{defn}
In the \textbf{Existential Forgery Problem} (EFP), we consider the following game:
\begin{enumerate}
\item A group $G\in \mathcal{G}$ and a key $k\in_R G$ are generated.
\item The adversary $\mathcal{A}$ gets the security parameter, in unary, and the group $G$.
\item $\mathcal{A}$ receives oracle access to the $E_k, D_k, P,$ and $P^{-1}$ oracles.
\item $\mathcal{A}$ eventually outputs a pair $(m,c)$.
\end{enumerate}
If $E_k(m)=c$, and $(m,c)$ has not been queried before, we say that $\mathcal{A}$ succeeds.

In the \textbf{Cracking Problem} (CP), we consider the following game:
\begin{enumerate}
\item A group $G\in \mathcal{G}$ and a key $k\in_R G$ are generated.
\item The adversary $\mathcal{A}$ gets the security parameter, in unary, and the group $G$.
\item $\mathcal{A}$ is presented with $E_k(m_0)=c_0\in_R G$.
\item $\mathcal{A}$ receives oracle access to the $E_k, D_k, P,$ and $P^{-1}$ oracles, but the decryption oracle outputs $\perp$ if $\mathcal{A}$ queries $c=c_0$.
\item $\mathcal{A}$ outputs a plaintext $m$.
\end{enumerate}
If $D_k(c_0)=m$, then we say that $\mathcal{A}$ succeeds. The \textbf{success probability} is the probability that on a uniformly random chosen encryption $c_0 = E_k(m_0)$, $\mathcal{A}$ outputs $m_0$.
\end{defn}

Furthermore, Even and Mansour show that polynomial-time EFP security infers poly\-nomial-time CP security. There are no limiting factors prohibiting the problems and inference result from being employed on groups. In fact, there is nothing disallowing the use of the same proof of the EFP security for the EFP security of the one-key EM scheme, as noted in \citep{DKS}, which we therefore omit. Indeed, by redefining notions in the \citep{EM} proof to take into account that we are working over a not necessarily abelian group, we are able to prove that the Group EM scheme satisfies the EFP notion of security, specifically the following.

\begin{thm}\label{MainEFP}
Assume $P\in_R\mathfrak{P}_{G\rightarrow G}$ and let the key $k\in_R G$. For any probabilistic adversary $\mathcal{A}$, the success probability of solving the EFP is bounded by
\begin{align*}
Succ(\mathcal{A}) = Pr_{k,P}\left[ EFP(\mathcal{A})=1\right] = O\left( \frac{st}{|G|} \right),
\end{align*}
where $s$ is the number of $E/D$-queries and $t$ is the number of $P/P^{-1}$-queries, i.e. the success probability is negligible.
\end{thm}

By the Even and Mansour inference result, we get the corollary below.

\begin{cor}
Assume $P\in_R\mathfrak{P}_{G\rightarrow G}$ and let the key $k\in_R G$. For any probabilistic polynomial-time (PPT) adversary $\mathcal{A}$, the success probability of solving the Cracking Problem is negligible.
\end{cor}

As Even and Mansour also note, the above results may be extended to instances where the permutation is a pseudorandom permutation by a simple reduction. Hence, we get the following two results.

\begin{thm}
Assume $P$ is a pseudorandom permutation on $G\in \mathcal{G}$ and let the key $k\in_R G$. For any probabilistic adversary $\mathcal{A}$ with only polynomially many queries to its oracles, the success probability of solving the Existential Forgery Problem is negligible.
\end{thm}

\begin{cor}
Assume $P$ is a pseudorandom permutation on $G\in \mathcal{G}$ and let the key $k\in_R G$. For any probabilistic polynomial-time (PPT) adversary $\mathcal{A}$, the success probability of solving the Cracking Problem is negligible.
\end{cor}

\subsection{Pseudorandomness Property of the Group EM Scheme}
Although the above notions of security are strong, we are more interested in any pseudorandomness property the Group EM scheme offers us. Kilian and Rogaway \citep{Kilr} show that the one-key EM scheme satisfies the pseudorandom permutation property, i.e. with only an encryption oracle and the permutation oracles, the EM scheme is indistinguishable from random to any adversary with only polynomially many queries to its oracles. We note that they only show the pseudorandomness property, but state in their discussion section that their proof may be adapted to include a decryption oracle, i.e. that the one-key EM scheme satisfies the super pseudorandom permutation property. Having done the analysis with the decryption oracle, over an arbitrary group, we concur. However, we were also able to generalize the \citep{Kilr} proof to a one-key construction. This not entirely remarkable as the key $k$ will usually be different from its group inverse, hence we were able to use the same proof, but with adjustments to the games and their analysis. 

In the following, we assume that the adversary $\mathcal{A}$ is unbounded computationally, but may only make polynomially many queries to the $E/D$- and $P/P^{-1}$-oracles, where all oracles act as black boxes and $P$ is a truly random permutation. We intend to play the "pseudorandom or random permutation game": $\mathcal{A}$ is given an encryption oracle $E$ (with related decryption oracle $D$) which is randomly chosen with equal probability from the following two options:
\begin{enumerate}
\item A random key $k\in_R G$ is chosen uniformly and used to encrypt as $E(m)=E_k(m)=P(m\cdot k )\cdot k$, or
\item A random permutation $\pi\in_R \mathfrak{P}_{G \rightarrow G}$ is chosen and used to encrypt as $E(m)=\pi(m)$.
\end{enumerate}
The adversary wins the game if it can distinguish how $E$ was chosen, with probability significantly better than $1/2$. More explicitly, we wish to prove the following for the group Even-Mansour scheme.

\begin{thm}\label{PseudoEMbounded}
Assume $P\in_R\mathfrak{P}_{G\rightarrow G}$ and let the key $k\in_R G$. For any probabilistic adversary $\mathcal{A}$, limited to polynomially many $E/D$- and $P/P^{-1}$-oracle queries, the adversarial advantage of $\mathcal{A}$ is bounded by
\begin{align}\label{AdvA}
\text{Adv}(\mathcal{A}) \defeq \left| Pr\left[ \mathcal{A}_{E_k,D_k}^{P,P^{-1}} = 1 \right] - Pr\left[\mathcal{A}_{\pi,\pi^{-1}}^{P,P^{-1}} = 1 \right]\right| = \mathcal{O}\left(\frac{st}{|G|}\right).
\end{align}
where $s$ is the total number of $E/D$-queries and $t$ is the total number of $P/P^{-1}$-queries, i.e. the success probability is negligible.
\end{thm}

\begin{proof}
We may assume that $\mathcal{A}$ is deterministic (in essence, being unbounded computationally affords $\mathcal{A}$ the possibility of derandomizing its strategy by searching all its possible random choices and picking the most effective choices after having computed the effectiveness of each choice. For an example, see \citep{DingDong}.) We may also assume that $\mathcal{A}$ never queries a pair in $S_s$ or $T_t$ more than once, where $S_i$ and $T_i$ are the sets of $i$ $E/D$- and $P/P^{-1}$-queries, respectively. Let us define two main games, that $\mathcal{A}$ could play, through oracle interactions (see next page for the explicit game descriptions.)

Note that the steps in italics have no impact on the response to $\mathcal{A}$'s queries, we simply continue to answer the queries and only note if the key turns bad, i.e. we say that a key $k$ is \textbf{bad w.r.t. the sets $S_s$ and $T_t$} if there exist $i,j$ such that either $m_i \cdot k = x_j$ or $c_i \cdot k^{-1} = y_j$, and $k$ is \textbf{good} otherwise. There are at most $\frac{2st}{|G|}$ bad keys.

\textbf{Game R}: We consider the random game which corresponds to the latter probability in (\ref{AdvA}), i.e.
\begin{align*}
P_R := Pr\left[ \mathcal{A}_{\pi,\pi^{-1}}^{P,P^{-1}}=1\right].
\end{align*}

 From the definition of \textbf{Game R}, we see that, letting $Pr_R$ denote the probability when playing \textbf{Game R},
\begin{align}\label{PR}
Pr_R \left[ \mathcal{A}_{E,D}^{P,P^{-1}}=1\right] = P_R,
\end{align}
as we are simply giving uniformly random answers to each of $\mathcal{A}$'s queries.

\newpage
\begin{footnotesize}\label{GamesXandR}
\noindent\textbf{Notation:} We let $S^1_i = \lbrace m | (m,c)\in S_{i} \rbrace, \hspace*{5pt} S^2_i = \lbrace c | (m,c)\in S_{i} \rbrace, T^1_i = \lbrace x | (x,y)\in T_{i} \rbrace,$ and $\hspace*{5pt} T^2_i = \lbrace y | (x,y)\in T_{i} \rbrace.$
\end{footnotesize}
\begin{small}
\hrule\noindent
\begin{minipage}[t]{0.48\textwidth}
\vspace*{0.01\textheight}\textbf{GAME R:} Initially, let $S_0$ and $T_0$ be empty and flag unset. Choose $k\in_R G$, then answer the $i+1$'st query as follows: \linebreak
\vspace*{0.0005\textheight}

\textbf{$E$-oracle query with $m_{i+1}$:} \\
 \textbf{1.} Choose $c_{i+1}\in_R G\setminus S^2_i$. \\
 \textbf{2.} \textit{If $P(m_{i+1}\cdot k)\in T^2_i$, or $P^{-1}(c_{i+1}\cdot k^{-1})\in T^1_i$, then set flag to \textbf{bad}.} \\
 \textbf{3.} Define $E(m_{i+1})=c_{i+1}$ (and thereby also $D(c_{i+1})=m_{i+1}$) and return $c_{i+1}$. \\

\vspace*{0.023\textheight}
\textbf{$D$-oracle query with $c_{i+1}$:} \\
 \textbf{1.} Choose $m_{i+1}\in_R G\setminus S^1_i$. \\
 \textbf{2.} \textit{If $P^{-1}(c_{i+1}\cdot k^{-1})\in T^1_i$, or $P(m_{i+1}\cdot k)\in T^2_i$, then set flag to \textbf{bad}.} \\
 \textbf{3.} Define $D(c_{i+1}) = m_{i+1}$ (and thereby also $E(m_{i+1})=c_{i+1}$) and return $m_{i+1}$. \\

\vspace*{0.045\textheight}
\textbf{$P$-oracle query with $x_{i+1}$:} \\
 \textbf{1.} Choose $y_{i+1}\in_R G\setminus T^2_i$. \\
 \textbf{2.} \textit{If $E(x_{i+1}\cdot k^{-1})\in S^2_i$, or $D(y_{i+1}\cdot k)\in S^1_i$, then set flag to \textbf{bad}.} \\
 \textbf{3.} Define $P(x_{i+1}) = y_{i+1}$ (and thereby also $P^{-1}(y_{i+1})=x_{i+1}$) and return $y_{i+1}$. \\

\vspace*{0.023\textheight}
\textbf{$P^{-1}$-oracle query with $y_{i+1}$:} \\
 \textbf{1.} Choose $x_{i+1}\in_R G\setminus T^1_i$. \\
 \textbf{2.} \textit{If $D(y_{i+1}\cdot k)\in S^1_i$, or $E(x_{i+1}\cdot k^{-1})\in S^2_i$, then set flag to \textbf{bad}.} \\
 \textbf{3.} Define $P^{-1}(y_{i+1}) = x_{i+1}$ (and thereby also $P(x_{i+1})=y_{i+1}$) and return $x_{i+1}$.
\end{minipage}\hspace*{0.01\textwidth} \vrule \hspace*{0.01\textwidth}
\begin{minipage}[t]{0.48\textwidth}
\vspace*{0.01\textheight}\textbf{GAME X:} Initially, let $S_0$ and $T_0$ be empty and flag unset. Choose $k\in_R G$, then answer the $i+1$'st query as follows: \linebreak
\vspace*{0.0005\textheight}

\textbf{$E$-oracle query with $m_{i+1}$:} \\
 \textbf{1.} Choose $c_{i+1}\in_R G\setminus S^2_i$. \\
 \textbf{2.} If $P(m_{i+1}\cdot k)\in T^2_i$ then redefine $c_{i+1} := P(m_{i+1}\cdot k)\cdot k$ \textit{and set flag to \textbf{bad}}. Else if $P^{-1}(c_{i+1}\cdot k^{-1})\in T^1_i$, \textit{then set flag to \textbf{bad} and} goto Step 1. \\
 \textbf{3.} Define $E(m_{i+1})=c_{i+1}$ (and thereby also $D(c_{i+1})=m_{i+1}$) and return $c_{i+1}$. \\

\textbf{$D$-oracle query with $c_{i+1}$:} \\
 \textbf{1.} Choose $m_{i+1}\in_R G\setminus S^1_i$. \\
 \textbf{2.} If $P^{-1}(c_{i+1}\cdot k^{-1})\in T^1_i$ then redefine $m_{i+1} := P^{-1}(c_{i+1}\cdot k^{-1})\cdot k^{-1}$ \textit{and set flag to \textbf{bad}}. Else if $P(m_{i+1}\cdot k)\in T^2_i$, \textit{then set flag to \textbf{bad} and} goto Step 1. \\
 \textbf{3.} Define $D(c_{i+1}) = m_{i+1}$ (and thereby also $E(m_{i+1})=c_{i+1}$) and return $m_{i+1}$. \\

\textbf{$P$-oracle query with $x_{i+1}$:} \\
 \textbf{1.} Choose $y_{i+1}\in_R G\setminus T^2_i$. \\
 \textbf{2.} If $E(x_{i+1}\cdot k^{-1})\in S^2_i$ then redefine $y_{i+1} := E(x_{i+1}\cdot k^{-1})\cdot k^{-1}$ \textit{and set flag to \textbf{bad}}. Else if $D(y_{i+1}\cdot k)\in S^1_i$, \textit{then set flag to \textbf{bad} and} goto Step 1. \\
 \textbf{3.} Define $P(x_{i+1}) = y_{i+1}$ (and thereby also $P^{-1}(y_{i+1})=x_{i+1}$) and return $y_{i+1}$. \\

\textbf{$P^{-1}$-oracle query with $y_{i+1}$:} \\
 \textbf{1.} Choose $x_{i+1}\in_R G\setminus T^1_i$. \\
 \textbf{2.} If $D(y_{i+1}\cdot k)\in S^1_i$ then redefine $x_{i+1} := D(y_{i+1}\cdot k)\cdot k$ \textit{and set flag to \textbf{bad}}. Else if $E(x_{i+1}\cdot k^{-1})\in S^2_i$, \textit{then set flag to \textbf{bad} and} goto Step 1. \\
 \textbf{3.} Define $P^{-1}(y_{i+1}) = x_{i+1}$ (and thereby also $P(x_{i+1})=y_{i+1}$) and return $x_{i+1}$.
\vspace*{0.02\textheight}
\end{minipage}
\end{small}

\newpage

\textbf{Game X}: Consider the experiment which corresponds to the game played in the prior probability in (\ref{AdvA}) and define this probability as
\begin{align*}
P_X := Pr\left[ \mathcal{A}_{E_k,D_k}^{P,P^{-1}}=1\right].
\end{align*}
We define \textbf{Game X}, as outlined above. Note that again the parts in italics have no impact on the response to $\mathcal{A}$'s queries, however, this time, when a key becomes \emph{bad}, we choose a new random value repeatedly for the response until the key is no longer \emph{bad}, and then reply with this value. Intuitively, \textbf{Game X} behaves like \textbf{Game R} except that \textbf{Game X} checks for consistency as it does not want $\mathcal{A}$ to win on some collision. It is non-trivial to see that, letting $Pr_X$ denote the probability when playing \textbf{Game X},
\begin{align}\label{PX}
Pr_X \left[ \mathcal{A}_{E,D}^{P,P^{-1}}=1\right] =  P_X.
\end{align}
The proof is given in Appendix~\ref{ExplainX}.

We have defined both games in such a way that their outcomes differ only in the event that a key turns \textit{bad}. Thus, any circumstance which causes a difference in the instructions carried out by the games, will also cause both games to set the flag to \emph{bad}. Let $BAD$ denote the event that the flag gets set to \emph{bad} and the case that the flag is not set to \emph{bad} by $\neg BAD$, then the two following lemmas follow from the previous statement.

\begin{lem}\label{BadisBad}
$Pr_R\left[ BAD \right] = Pr_X\left[ BAD \right]$ and $Pr_R\left[ \neg BAD \right] = Pr_X\left[ \neg BAD \right]$.
\end{lem}

\begin{lem}\label{Notbadisnotbad}
$Pr_R\left[ \mathcal{A}_{E,D}^{P,P^{-1}}=1| \neg BAD \right] = Pr_X\left[ \mathcal{A}_{E,D}^{P,P^{-1}}=1| \neg BAD \right]$.
\end{lem}

Using these two lemmas we are able to prove the lemma:

\begin{lem}
$\text{Adv}(\mathcal{A}) \leq Pr_R\left[ BAD \right]$.
\end{lem}

This is because, using (\ref{PR}), (\ref{PX}), and lemmas \ref{BadisBad} and \ref{Notbadisnotbad}, 
\begin{align*}
\text{Adv}(\mathcal{A}) &= |P_X - P_R| \\
						&= \left| Pr_X\left[ \mathcal{A}_{E,D}^{P,P^{-1}}=1 \right] - Pr_R\left[ \mathcal{A}_{E,D}^{P,P^{-1}}=1 \right] \right| \\
						&= | Pr_X\left[ \mathcal{A}_{E,D}^{P,P^{-1}}=1 | \neg BAD \right]\cdot Pr_X\left[\neg BAD \right] \\
							&\hspace*{25pt}+ Pr_X\left[ \mathcal{A}_{E,D}^{P,P^{-1}}=1 | BAD \right]\cdot Pr_X\left[ BAD \right] \\
							&\hspace*{45pt} - Pr_R\left[ \mathcal{A}_{E,D}^{P,P^{-1}}=1 | \neg BAD \right]\cdot Pr_R\left[\neg BAD \right] \\
							&\hspace*{65pt}- Pr_R\left[ \mathcal{A}_{E,D}^{P,P^{-1}}=1 | BAD \right]\cdot Pr_R\left[ BAD \right] | \\
						&= \left| Pr_R\left[ BAD \right]\cdot \left( Pr_X\left[ \mathcal{A}_{E,D}^{P,P^{-1}}=1 | BAD \right] - Pr_R\left[ \mathcal{A}_{E,D}^{P,P^{-1}}=1 | BAD \right]\right) \right| \\
						&\leq Pr_R\left[ BAD \right].
\end{align*}

Let us now define yet another game, \textbf{Game R'}.
\vspace*{0.015\textheight}
\hrule
\begin{footnotesize}
\begin{minipage}[t]{0.9\textwidth}
\vspace*{0.005\textheight} \textbf{GAME R':} Initially, let $S_0$ and $T_0$ be empty and flag unset. Answer the $i+1$'st query as follows: \linebreak
\textbf{$E$-oracle query with $m_{i+1}$:} \\
 \textbf{1.} Choose $c_{i+1}\in_R G\setminus S^2_i$. \\
 \textbf{2.} Define $E(m_{i+1}):=c_{i+1}$ (and thereby also $D(c_{i+1}):=m_{i+1}$) and return $c_{i+1}$. \\

\textbf{$D$-oracle query with $c_{i+1}$:} \\
 \textbf{1.} Choose $m_{i+1}\in_R G\setminus S^1_i$. \\
 \textbf{2.} Define $D(c_{i+1}) := m_{i+1}$ (and thereby also $E(m_{i+1}):=c_{i+1}$) and return $m_{i+1}$. \\

\textbf{$P$-oracle query with $x_{i+1}$:} \\
 \textbf{1.} Choose $y_{i+1}\in_R G\setminus T^2_i$. \\
 \textbf{2.} Define $P(x_{i+1}) := y_{i+1}$ (and thereby also $P^{-1}(y_{i+1}):=x_{i+1}$) and return $y_{i+1}$. \\

\textbf{$P^{-1}$-oracle query with $y_{i+1}$:} \\
 \textbf{1.} Choose $x_{i+1}\in_R G\setminus T^1_i$. \\
 \textbf{2.} Define $P^{-1}(y_{i+1}) := x_{i+1}$ (and thereby also $P(x_{i+1}):=y_{i+1}$) and return $x_{i+1}$.\\
 
\textit{After all queries have been answered, choose $k\in_R G$. If there exists $(m,c)\in S_{s}$ and $(x,y)\in T_{t}$ such that $k$ becomes bad then set flag to \textbf{bad}.}
\vspace*{0.005\textheight}
\end{minipage}
\end{footnotesize}
\hrule
\vspace*{0.015\textheight}
This game runs as \textbf{Game R} except that it does not choose a key until all of the queries have been answered and then checks for badness of the flag (by checking whether or not the key has become bad). It can be shown that the flag is set to \textbf{\textit{bad}} in \textbf{Game R} if and only if the flag is set to \textbf{\textit{bad}} in \textbf{Game R'} (by a consideration of cases (see Appendix~\ref{App:ReqRR}.)) Hence, we get the following lemma.

\begin{lem}
$Pr_R\left[ BAD \right] = Pr_{R'} \left[ BAD \right]$.
\end{lem}

Using the above lemma, we now only have to bound $Pr_{R'} \left[ BAD \right]$ in order to bound $\text{Adv}(\mathcal{A})$, but as the adversary queries at most $s$ elements to the $E/D$-oracles and at most $t$ elements to the $P/P^{-1}$-oracles, and the key $k$ is chosen uniformly at random from $G$, we have that the probability of choosing a bad key is at most $2st/|G|$, i.e.
\begin{align*}
\text{Adv}(\mathcal{A}) \leq Pr_{R'} \left[ BAD \right] = \mathcal{O}\left( \frac{st}{|G|} \right).
\end{align*}

\end{proof}

Restating the theorem, we get our second main result:

\begin{thm}
For any probabilistic adversary $\mathcal{A}$, limited to polynomially many $E/D$- and $P/P^{-1}$-oracle queries, the Group EM scheme over a group $G$ is a super pseudorandom permutation.
\end{thm}

By removing the decryption oracle, we get the following corollary:

\begin{cor}
For any probabilistic adversary $\mathcal{A}$, limited to polynomially many $E$- and $P/P^{-1}$-oracle queries, the Group EM scheme over a group $G$ is a pseudorandom permutation.
\end{cor}

\begin{rem}
We see that in the group $((\mathbb{Z}/2\mathbb{Z})^n,\oplus)$, our Group EM scheme reduces to the one-key EM scheme given in \citep{DKS}. The proof given in \citep{DKS} proves the security of the scheme, and the proof given in \citep{Kilr} proves the pseudorandomness, equivalently to our claims.
\end{rem}

It can be proven that a multiple round Group EM scheme is an SPRP because the security only depends on the last round, which is also an SPRP.

\subsection{Slide Attack}
We would like to show that the security bound that we have found above is optimal, so we slightly alter the simple optimal attack on the Single-Key Even-Mansour cipher as constructed in \citep{DKS}. The original version works for abelian groups with few adjustments and \citep{DKS} also present another slide attack against a modular addition DESX construction.

Consider the one-key Group Even-Mansour cipher
\begin{align*}
E(x) = P(x \cdot k) \cdot k,
\end{align*}
over a group $G$ with binary operation $\cdot$, where $P$ is a publicly available permutation oracle, $x\in G$, and $k\in_R G$. Define the following values:
\begin{align*}
x=x, \hspace*{5pt} y=x \cdot k, \hspace*{5pt} z= P(y), \hspace*{5pt} w=E(x)=P(x \cdot k) \cdot k.
\end{align*}
We hereby have that $w\cdot y^{-1} = z \cdot x^{-1} $. Consider the attack which follows.
\begin{enumerate}
\item For $d = \sqrt{|G|}$ arbitrary values $x_i\in G$, $i=1,\ldots, d$, and $d$ arbitrary values $y_i\in G$, $i=1,\ldots, d$, query the $E$-oracle on the $x_i$'s and the $P$-oracle on the $y_i$'s. Store the values in a hash table as
\begin{align*}
(E(x_i)\cdot y_i^{-1}, P(y_i)\cdot x_i^{-1}, i),
\end{align*}
sorted by the first coordinate.
\item If there exists a match in the above step, i.e. $E(x_i)\cdot y_i^{-1} = P(y_i)\cdot x_i^{-1}$ for some $i$, check the guess that $k = x_i^{-1}\cdot y_i$.
\end{enumerate}
It can be seen by the Birthday Problem\footnote{Considering the approximation $p(n)\approx \tfrac{n^2}{2m}$, where $p(n)$ is the probability of there being a Birthday Problem collision from $n$ randomly chosen elements from the set of $m$ elements, then $p(\sqrt{|G|})\approx \tfrac{\sqrt{|G|}^2}{2|G|}=1/2$.}, that with non-negligible probability, there must exist a slid pair $(x_i,y_i)$ satisfying the above property, i.e. there exists $1\leq i \leq d$ such that $k = x_i^{-1} \cdot y_i$. For a random pair $(x,y)\in G^2$ it holds that $E(x) = P(y) \cdot x^{-1} \cdot y$ with probability $|G|^{-1}$, so we expect few, if any, collisions in the hash table, including the collision by the slid pair where the correct key $k$ is found. The data complexity of the attack is $d$ $E$-oracle queries and $d$ $P$-oracle queries. Hence the attack bound $d^2 = |G|$, which matches the lower bound given in Theorem~\ref{MainEFP} and Theorem~\ref{PseudoEMbounded}. We have therefore found that our scheme is optimal.

\newpage
\section{Feistel}\label{FeistelSection}
We now consider the Feistel cipher over arbitrary groups, which we will call the Group Feistel cipher. The following is a complement to \citep{PatelRamzanSundaram} who treat the Group Feistel cipher construction with great detail. Our main accomplishment in this section is the settling of an open problem posed by them.

\subsection{Definitions}
We define a Feistel cipher over a group $(G,\cdot)$ as a series of round functions on elements of $G\times G=G^2$.

\begin{defn}
Given an efficiently computable but not necessarily invertible function $f: G \rightarrow G$, called a \textbf{round function}, we define the \textbf{1-round Group Feistel cipher} $\mathcal{F}_{f}$ to be
\begin{align*}
\mathcal{F}_{f}: 	G \times G &\longrightarrow G \times G,\\
				(x,y) &\longmapsto (y, x \cdot f(y)).
\end{align*}
In the case where we have multiple rounds, we index the round functions as $f_i$, and denote the \textbf{$r$-round Group Feistel cipher} by $\mathcal{F}_{f_1,\ldots,f_r}$. We concurrently denote the input to the $i$'th round by $(L_{i-1},R_{i-1})$ and having the output $(L_i,R_i) = (R_{i-1}, L_{i-1}\cdot f_i(R_{i-1}))$, where $L_i$ and $R_i$ respectively denote the left and right parts of the $i$'th output.
\end{defn}

Note that if $(L_i,R_i)$ is the $i$'th round output, we may invert the $i$'th round by setting $R_{i-1}:=L_i$ and then computing $L_{i-1}:= R_i \cdot (f_i(R_{i-1}))^{-1}$ to get $(L_{i-1},R_{i-1})$. As this holds for all rounds, regardless of the invertibility of the round functions, we get that an $r$-round Feistel cipher is invertible for all $r$.

Let $F:G_\lambda \times G \rightarrow G$ be a pseudorandom function. We define the keyed permutation $F^{(r)}$ as
\begin{align*}
F^{(r)}_{k_1,\ldots,k_r}(x,y) \defeq \mathcal{F}_{F_{k_1},\ldots, F_{k_r}}(x,y).
\end{align*}
We sometimes index the keys as $1,2,\ldots, r$, or omit the key index entirely.

\subsection{Results}
For completeness, we show some of the preliminary results for Group Feistel ciphers, not considered in \citep{PatelRamzanSundaram}.

We first note that $F^{(1)}$ is \textit{not} a pseudorandom permutation as 
\begin{align*}
F^{(1)}_{k_1}(L_0,R_0) = (L_1,R_1) = (R_0,L_0\cdot F_{k_1}(R_0)),
\end{align*}
such that any distinguisher $\mathcal{A}$ need only compare $R_0$ to $L_1$.

Also $F^{(2)}$ is \textit{not} a pseudorandom permutation: Consider a pseudorandom function $F$ on $G$. Pick $k_1,k_2\in_R G_\lambda$. Distinguisher $\mathcal{A}$ sets $(L_0,R_0)=(1,g)$ for some $g\in G$, where $1$ is the identity element of $G$, then queries $(L_0,R_0)$ to its oracle and receives,
\begin{center}
$L_2 = L_0 \cdot F_{k_1}(R_0) = F_{k_1}(g)$ and $R_2 = R_0 \cdot F_{k_2}(L_0\cdot F_{k_1}(R_0)) = g \cdot F_{k_2}(F_{k_1}(g))$.
\end{center}
On its second query, the distinguisher $\mathcal{A}$ lets $L_0 \in G \setminus {1}$ but $R_0=g$, such that it receives
\begin{center}
$L_2 = L_0 \cdot F_{k_1}(R_0) = L_0 \cdot F_{k_1}(g)$ and $R_2 = g \cdot F_{k_2}(L_0 \cdot F_{k_1}(g))$.
\end{center}
As $\mathcal{A}$ may find the inverse to elements in $G$, $\mathcal{A}$ acquires $(F_{k_1}(g))^{-1}$, and by so doing, may compute $L_2 \cdot (F_{k_1}(g))^{-1} = L_0$. If $F^{(2)}$ were random, this would only occur negligibly many times, while $\mathcal{A}$ may query its permutation-oracle polynomially many times such that if $L_0$ is retrieved non-negligibly many times out of the queries, $\mathcal{A}$ is able to distinguish between a random permutation and $F^{(2)}$ with non-negligible probability.

As one would expect, the $3$-round Group Feistel cipher (see Figure~\ref{3roundFeistel}) is indeed a pseudorandom permutation. The following proof is based on the proof given in Katz and Lindell \citep{KL} of the analogous result for bit strings with XOR.

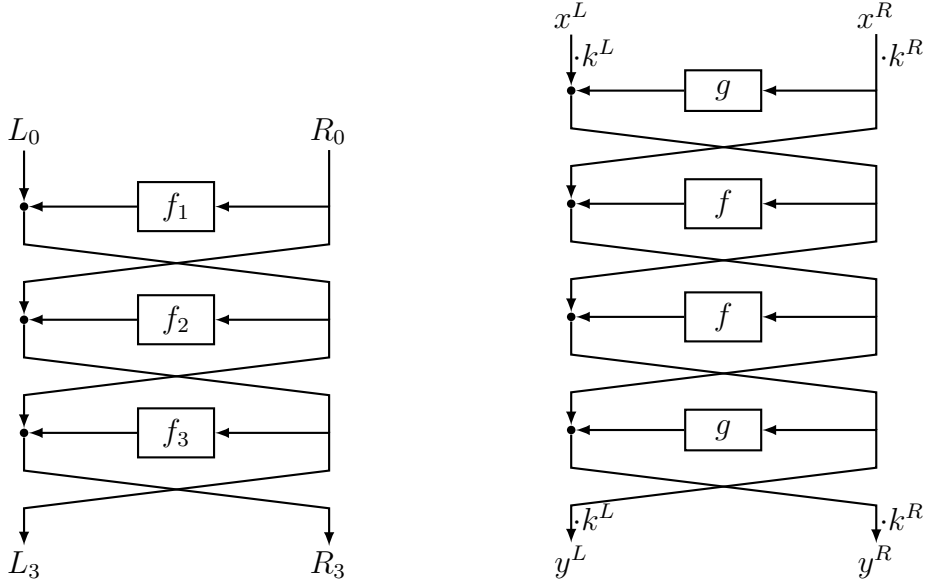
\begin{figure}
\centering
\begin{subfigure}[b]{0.49\textwidth}
\centering
\begin{tikzpicture}

    \tikzstyle{dot} = [
		fill,
		shape=circle,
		minimum size=4pt,
		inner sep=0pt,
	]

    \foreach \z in {1, 2,...,3} {
        \node[draw,thick,minimum width=1cm] (g\z) at ($\z*(0,-1.5cm)$)  {$f_\z$};
        \node (prik\z) [dot, left of = g\z, node distance = 2cm, scale=0.8] {};
        \draw[thick,-latex] (g\z) -- (prik\z);
    
    }
    
    \foreach \z in {1, 2} {
   	 	\draw[thick,latex-latex] (g\z.east) -| +(1.5cm,-0.5cm) -- ($(prik\z) - (0,1cm)$) -- ($(prik\z.north) - (0,1.5cm)$);
   	 	\draw[thick] (prik\z.south) -- ($(prik\z)+(0,-0.5cm)$) -- ($(g\z.east) + (1.5cm,-1cm)$) -- +(0,-0.5cm);
    }

	%% Inputs    
    \node (p0) [above of = g1, minimum width=5cm,minimum height=0.5cm,node distance=1cm] {}; 
    \node (l0) [above of = prik1,node distance=1cm] {$L_0$};
    \node (r0) [right of = l0, node distance = 4cm] {$R_0$};
    \draw[thick,-latex] (l0 |- p0.south) -- (prik1.north);
    \draw[thick] ($(g1.east)+(1.5cm,0)$) -- +(0,0.75cm);

	%% Outputs
    \node (p3) [below of = g3, minimum width=5cm,minimum height=0.5cm,node distance=1.75cm] {}; 
    \node (l3) [below of = prik3,node distance=1.75cm] {$L_3$};
    \node (r3) [right of = l3, node distance = 4cm] {$R_3$};
    \draw[thick,latex-latex] (g3.east) -| +(1.5cm,-0.5cm) -- ($(prik3) - (0,1cm)$) -- (prik3 |- p3.north);
    \draw[thick,-latex] (prik3.south) -- ($(prik3)+(0,-0.5cm)$) -- ($(g3.east) + (1.5cm,-1cm)$) -- +(0,-0.5cm);

\end{tikzpicture}
\captionof{figure}[$3$-round Group Feistel cipher.]{$3$-round Group Feistel cipher.}
\label{3roundFeistel}
\end{subfigure}
\begin{subfigure}[b]{0.49\textwidth}
\centering
\begin{tikzpicture}

    \tikzstyle{dot} = [
		fill,
		shape=circle,
		minimum size=4pt,
		inner sep=0pt,
	]

%%    \foreach \z in {1, 2,...,4} {
%%        \node[draw,thick,minimum width=1cm] (f\z) at ($\z*(0,-1.5cm)$)  {};
%%        \node (xor\z) [dot, left of = f\z, node distance = 2cm, scale=0.8] {};
%%        \draw[thick,-latex] (f\z) -- (xor\z);
    
    \node[draw,thick,minimum width=1cm] (f1) at ($1*(0,-1.5cm)$)  {$g$};
    \node (xor1) [dot, left of = f1, node distance = 2cm, scale=0.8] {};
    \draw[thick,-latex] (f1) -- (xor1);   
    \node[draw,thick,minimum width=1cm] (f2) at ($2*(0,-1.5cm)$)  {$f$};
    \node (xor2) [dot, left of = f2, node distance = 2cm, scale=0.8] {};
    \draw[thick,-latex] (f2) -- (xor2);   
    \node[draw,thick,minimum width=1cm] (f3) at ($3*(0,-1.5cm)$)  {$f$};
    \node (xor3) [dot, left of = f3, node distance = 2cm, scale=0.8] {};
    \draw[thick,-latex] (f3) -- (xor3);   
    \node[draw,thick,minimum width=1cm] (f4) at ($4*(0,-1.5cm)$)  {$g$};
    \node (xor4) [dot, left of = f4, node distance = 2cm, scale=0.8] {};
    \draw[thick,-latex] (f4) -- (xor4);

    \foreach \z in {1, 2,...,3} {
   	 	\draw[thick,latex-latex] (f\z.east) -| +(1.5cm,-0.5cm) -- ($(xor\z) - (0,1cm)$) -- ($(xor\z.north) - (0,1.5cm)$);
   	 	\draw[thick] (xor\z.south) -- ($(xor\z)+(0,-0.5cm)$) -- ($(f\z.east) + (1.5cm,-1cm)$) -- +(0,-0.5cm);
    }

	%% Inputs    
    \node (p0) [above of = f1, minimum width=5cm,minimum height=0.5cm,node distance=1cm] {}; 
    \node (l0) [above of = xor1,node distance=1cm] {$x^L$};
    \node (mid0) [above of = xor1,node distance = .5cm] {};
    \node (kl0) [right of = mid0,node distance=.3cm] {$\cdot k^L$};
    \node (r0) [right of = l0, node distance = 4cm] {$x^R$};
    \node (kr0) [right of = kl0,node distance= 4.05cm] {$\cdot k^R$};
    \draw[thick,-latex] (l0 |- p0.south) -- (xor1.north);
    \draw[thick] ($(f1.east)+(1.5cm,0)$) -- +(0,0.75cm);

	%% Outputs
    \node (p4) [below of = f4, minimum width=5cm,minimum height=0.5cm,node distance=1.75cm] {}; 
    \node (l4) [below of = xor4,node distance=1.75cm] {$y^L$};
    \node (kl4) [below of = kl0,node distance=6.15cm] {$\cdot k^L$};
    \node (r4) [right of = l4, node distance = 4cm] {$y^R$};
    \node (kr4) [right of = kl4,node distance= 4.05cm] {$\cdot k^R$};
    \draw[thick,latex-latex] (f4.east) -| +(1.5cm,-0.5cm) -- ($(xor4) - (0,1cm)$) -- (xor4 |- p4.north);
    \draw[thick,-latex] (xor4.south) -- ($(xor4)+(0,-0.5cm)$) -- ($(f4.east) + (1.5cm,-1cm)$) -- +(0,-0.5cm);

\end{tikzpicture}
\captionof{figure}[Group EM scheme with Feistel.]{Group EM scheme with Feistel.\footnotemark}
\label{GendGenRamPic}
\end{subfigure}
\captionof{figure}{Encryption schemes.}
\end{figure}

\footnotetext{TikZ figure adapted from \citep{TikZhelp}.}

\begin{thm}
If $F$ is a pseudorandom function on $G$, then $F^{(3)}$ is a pseudorandom permutation on $G$.
\end{thm}

\begin{proof}
Assume we have chosen uniform keys $k_1, k_2, k_3\in G$. We replace our pseudorandom functions in $F^{(3)}$ by random functions $f_i$ for $i=1,2,3$, which we may by a standard argument: Simply consider an adversary using such a distinguisher as a subroutine. The adversary has an oracle-triple of either three random functions or three independently keyed pseudorandom functions and the reduction easily follows.

Let us continue our proof by letting $\mathcal{A}$ be a probabilistic distinguisher with only polynomially many queries to its permutation-oracle, which is either $F^{(3)}$ or $\pi\in_R \mathfrak{P}_{G^2\rightarrow G^2}$, with equal probability. We wish to show that
\begin{align*}
\left| \underset{f_1,f_2,f_3\in_R \mathfrak{F}_{G\rightarrow G}}{Pr}\left[ \mathcal{A}^{\mathcal{F}_{f_1,f_2,f_3}(\cdot)}(n)=1 \right] - \underset{\pi \in_R \mathfrak{P}_{G^2 \rightarrow G^2}}{Pr}\left[ \mathcal{A}^{\pi(\cdot)}(n)=1 \right] \right|,
\end{align*}
is negligible over uniform and independent choices of random functions $f_1,f_2,$ and $f_3$ in $\mathfrak{F}_{G\rightarrow G}$, where the $n\in \mathbb{N}$ denotes the security parameter (written in unary, which decides the choice of $G$.) We assume that $q(n)=q$ is the polynomial bound on the permutation-oracle, $\mathcal{O}$, queries that $\mathcal{A}$ petitions and let $Q=\lbrace 1,\ldots, q\rbrace$. We may assume that $\mathcal{A}$ never queries the same pair to the $\mathcal{O}$-oracle twice. We let $(L_0^i,R_0^i),(L_1^i,R_1^i),(L_2^i,R_2^i),$ and $(L_3^i,R_3^i)$ denote the respectively $0$th, $1$st, $2$nd, and $3$rd intermediate values of the $i$th query. In this notation, $\mathcal{A}$ only ever sees the corresponding $(L_0^i,R_0^i)$ and $(L_3^i,R_3^i)$ pairs of values upon oracle query.
We wish to prove that the Feistel cipher is collision-free, such that the distinguisher cannot win on there being a collision.

\textbf{Collision at $R_1$:}
If $R_1^i = R_1^j$ for some $i\neq j \in Q$, we say that there is a \textit{collision at $R_1$}. We now consider when such a collision might occur. Assume there exist fixed, distinct $i,j\leq q$ such that $R_0^i = R_0^j$, then, as we assumed $\mathcal{A}$ does not query the same pair twice, $L_0^i \neq L_0^j$, such that
\begin{align*}
R_1^i = L_0^i \cdot f_1(R_0^i) \neq L_0^j \cdot f_1(R_0^j) = R_1^j,
\end{align*}
i.e. there is no collision at $R_1$. If instead $R_0^i \neq R_0^j$, then $f_1(R_0^i)$ and $f_1(R_0^j)$ are uniformly distributed and independent, as we assumed that $f_1$ was random, and so
\begin{align*}
Pr\left[ L_0^i \cdot f_1(R_0^i) = L_0^j \cdot f_1(R_0^j) \right] = Pr\left[ f_1(R_0^j) = (L_0^j)^{-1}\cdot L_0^i \cdot f_1(R_0^i) \right] = \frac{1}{|G|}.
\end{align*}
By the union bound over the $q$ values of $i$ and $q$ values of $j$ for distinct $i,j\in Q$, we have that,
\begin{align*}
Pr\left[ \text{ Collision at } R_1 \right] \leq q^2/|G|,
\end{align*}
which, for large enough $|G|$, is negligible.

\textbf{Collision at $R_2$:}
As with the case of $R_1$, if $R_2^i = R_2^j$ for some $i\neq j\in Q$, we say that there is a \textit{collision at $R_2$}. Obviously, if we condition a collision at $R_2$ by a collision at $R_1$, then the probability of a collision at $R_2$ must be negligible as a collision at $R_1$ is negligible. Let us therefore condition a collision at $R_2$ on there being \textit{no} collision at $R_1$, we wish to prove that also this probability is negligible. Consider again fixed and distinct $i,j\in Q$. By the assumption of no collision at $R_1$, we have $R_1^i \neq R_1^j$. Hence $f_2(R_1^i)$ and $f_2(R_1^j)$ are uniformly distributed and independent because we assumed that $f_2$ was random. Furthermore, $f_1$ and $f_2$ are independent, so one easily sees that
\begin{align}\label{CollatR2}
Pr\left[ L_1^i \cdot f_2(R_1^i) = L_1^j \cdot f_2(R_1^j) \mid \text{No collision at } R_1 \right] = \frac{1}{|G|},
\end{align}
such that again, by taking the union bound over all distinct $i,j\in Q$,
\begin{align*}
Pr\left[ \text{ Collision at } R_2 \mid \text{No collision at } R_1 \right] \leq q^2/|G|,
\end{align*}
which, for large enough $|G|$, is negligible.

Now, conditioned on there being \textit{no} collision at $R_1$, for all distinct $i,j\in Q$, we have that $L_3^i = R_2^i = L_1^i \cdot f_2(R_1^i)$ is independent of $L_3^j$, and uniformly distributed in $G$ (see (\ref{CollatR2})). Additionally, conditioning on there being \textit{no} collision at $R_2$, we get that $L_3^1,\ldots,L_3^q$ are uniformly distributed among all sequences of distinct values in $G$. Again conditioned on there being \textit{no} collision at $R_2$ and using that $L_2^i$ is independent of $L_3^i$ and that $f_3$ is random and independent over uniform and independent input, we similarly see that for $R_3^i = L_2^i \cdot f_3(R_2^i)$, the $q$ values $R_3^1, \ldots, R_3^q$ are uniformly distributed in $G$ and independent of each other, as well as independent of the $L_3^1,\ldots, L_3^q$.

So, when querying the $F^{(3)}$ instantiated $\mathcal{O}$-oracle (with uniform and independent round functions) with $q$ distinct inputs, the outputs $(L_3^1,R_3^1),\ldots,(L_3^q,R_3^q)$ are, except with negligible probability, distributed such that the $\lbrace L_3^i \rbrace$ are uniform and independent elements of $G$ and also the $\lbrace R_3^i \rbrace$ are uniform and independent elements of $G$. However, if the $\mathcal{O}$-oracle is instantiated with a random permutation, the outputs $(L_3^1,R_3^1),\ldots,(L_3^q,R_3^q)$ are uniform and independent elements of $G\times G$. Hence, the distinguisher's best attack would be to guess that the $\mathcal{O}$-oracle is a random permutation if there exist distinct $i,j \in Q$ such that $L_3^i = L_3^j$, which only occurs with negligible probability. Specifically, for a $3$-round Feistel network, $F^{(3)}$, using uniform and independent round functions, there exists a negligible function $negl_1(\cdot)$, such that
\begin{align*}
\underset{F^{(3)}}{Pr}\left[ (L_3^i,R_3^i) = (L_3^j,R_3^j) | i\neq j\in Q \right] \leq negl_1(n),
\end{align*}
for large enough $n$, and for a random permutation $\pi$, there exists a negligible function $negl_2(\cdot)$, such that
\begin{align*}
\underset{\pi}{Pr}\left[ (L_3^i,R_3^i) = (L_3^j,R_3^j) | i\neq j\in Q \right] = \frac{q\cdot (q-1)}{|G|} \leq negl_2(n),
\end{align*}
for large enough $n$. But then, there exists a negligible function $negl(\cdot)$, such that
\begin{footnotesize}
\begin{align*}
\left| \underset{\mathcal{F}}{Pr}\left[ (L_3^i,R_3^i) = (L_3^j,R_3^j) | i\neq j\in Q \right] - \underset{\pi}{Pr}\left[ (L_3^i,R_3^i) = (L_3^j,R_3^j) | i\neq j\in Q \right] \right| \leq negl(n),
\end{align*}
\end{footnotesize}
for large enough $n$. As any collision based attack will not give the distinguisher a non-negligible advantage $F^{(3)}$ must be a pseudorandom permutation.
\end{proof}

Among the considerations in \citep{PatelRamzanSundaram}, they showed that the $3$-round Feistel cipher over abelian groups was not super pseudorandom, but left as an open problem a proof over non-abelian groups. We present such a proof now.

\begin{prop}
The $3$-round Group Feistel cipher is not super pseudorandom.
\end{prop}

\begin{proof}
The proof is a counter-example using the following procedure:
\begin{enumerate}
\item Choose two oracle-query pairs in $G\times G$: $(L_0,R_0)$ and $(L'_0,R_0)$ where $L_0\neq L'_0$.
\item Query the encryption oracle to get $(L_3,R_3)$ and $(L'_3,R'_3)$.
\item Query $(L''_3,R''_3)= (L'_3,L_0\cdot (L'_0)^{-1} \cdot R'_3)$ to the decryption oracle.
\item If $R''_0=L'_3\cdot (L_3)^{-1} \cdot R_0$, guess that the oracle is $F^{(3)}$, else guess random.
\end{enumerate}
For $F^{(3)}$, this algorithm succeeds with probability $1$. For a random permutation, this algorithm succeeds negligibly often.

Explicitly, assume the $\mathcal{O}$-oracle is instantiated using $F^{(3)}$ with the three pseudorandom functions $F_1,F_2,F_3$. If we choose the two $\mathcal{O}$-oracle query pairs in $G\times G$: $(L_0,R_0)$ and $(L'_0,R'_0) = (L'_0,R_0)$ where $L_0 \neq L'_0$, then we get that
\begin{align*}
(L_0,R_0) &\overset{\mathcal{O}}\longmapsto (L_3,R_3)=(R_0\cdot F_2(L_0 \cdot F_1(R_0)),L_0\cdot F_1(R_0)\cdot F_3(L_3)) \\
(L'_0,R_0) &\overset{\mathcal{O}}\longmapsto (L'_3,R'_3) = (R_0\cdot F_2(L'_0 \cdot F_1(R_0)),L'_0\cdot F_1(R_0)\cdot F_3(L'_3)).
\end{align*}
Take the $\mathcal{O}^{-1}$-oracle query pair in $G\times G$:
\begin{align*}
(L''_3,R''_3) = (L'_3,L_0\cdot (L'_0)^{-1} \cdot R'_3),
\end{align*} 
which is different from $L'_3,R'_3$ as $L_0 \neq L'_0$ such that $L_0\cdot (L'_0)^{-1} \neq 1$. We then get the following by following the steps of the $\mathcal{O}^{-1}$-oracle.
\begin{align*}
L''_3 &:= L'_3 = R_0 \cdot F_2(L'_0 \cdot F_1(R_0)), \\
R''_3 &:= L_0 \cdot (L'_0)^{-1} \cdot R'_3 \\
	&= L_0 \cdot (L'_0)^{-1} \cdot L'_0\cdot F_1(R_0)\cdot F_3(L'_3), \\
	&= L_0\cdot F_1(R_0)\cdot F_3(L'_3).
\end{align*}
As $R''_2 := L''_3=L'_3$ and $L''_2 := R''_3 \cdot (F_3(R''_2))^{-1}$, we get that
\begin{align*}
L''_2 &:= R''_3 \cdot (F_3(R''_2))^{-1} \\
	&= L_0\cdot F_1(R_0)\cdot F_3(L'_3) \cdot (F_3(L'_3))^{-1} \\
	&= L_0\cdot F_1(R_0).
\end{align*}
Using that $L''_1 := R''_2 \cdot (F_2(R''_1))^{-1}$, $R''_2 = L'_3$, and $R''_1 := L''_2$, we get that
\begin{align*}
L''_1 &:= R''_2 \cdot (F_2(R''_1))^{-1} \\
	&= L'_3 \cdot (F_2(L_0\cdot F_1(R_0)))^{-1} \\
	&= R_0 \cdot F_2(L'_0 \cdot F_1(R_0)) \cdot (F_2(L_0\cdot F_1(R_0)))^{-1}.
\end{align*}
Using that $R_0$ is known such that we may find its inverse, we hereby get that
\begin{align*}
R''_0 := L''_1 &= R_0 \cdot F_2(L'_0 \cdot F_1(R_0)) \cdot (F_2(L_0\cdot F_1(R_0)))^{-1} \\
			&= R_0 \cdot F_2(L'_0 \cdot F_1(R_0)) \cdot (F_2(L_0\cdot F_1(R_0)))^{-1} \cdot R_0^{-1} \cdot R_0 \\
			&= R_0 \cdot F_2(L'_0 \cdot F_1(R_0)) \cdot (R_0 \cdot F_2(L_0\cdot F_1(R_0)))^{-1} \cdot R_0 \\
			&= L'_3 \cdot (L_3)^{-1} \cdot R_0.
\end{align*}
The elements $L'_3, L_3,$ and $R_0$ are all known values, and so, by querying the $\mathcal{O}$-oracle and its inverse oracle with the above queries, we may determine $R''_0$ with probability $1$ if the oracle is the Feistel network, and with negligible probability if the inverse oracle were random. Hence, we have demonstrated an attack on $F^{(3)}$ such that it cannot be a super pseudorandom permutation over any group $G$.
\end{proof}

For super pseudorandomness of the $4$-round Group Feistel cipher, we refer the reader to \citep{PatelRamzanSundaram}. In the paper, they show a strong result using certain hash functions as round functions, from which the following is a corollary.

\begin{cor}
Let $G$ be a group, with characteristic other than $2$, and let $f,g: G_\lambda \times G \rightarrow G$ be pseudorandom functions. Then, for any adversary $\mathcal{A}$ with polynomially many queries to its $E/D$-oracles, the family $\mathcal{P}$ of permutations on $G\times G$ consisting of permutations of the form $F^{(4)}=\mathcal{F}_{g,f,f,g}$ are indistinguishable from random, i.e. super pseudorandom permutations (SPRPs).
\end{cor}

\newpage
\section{Implementing the Group Even-Mansour Scheme}\label{ImplementEM}
Now that we have shown that both the Even-Mansour scheme and the Feistel cipher are generalizable to arbitrary groups, we might consider how to implement one given the other. Gentry and Ramzan \citep{GentryRamzan} considered exactly this for the two-key EM scheme over $(\mathbb{Z}/2\mathbb{Z})^n$. However, their paper only had sketches of proofs and refer to another edition of the paper for full details. As we are unable to find a copy in the place that they specify it to exist, and as we generalize their result non-trivially, we have decided to fill in the details while generalizing their proof.

In this section, we consider a generalized version of the Gentry and Ramzan \citep{GentryRamzan} construction, namely, the Group Even-Mansour scheme on $G^2$ instantiated with a $4$-round Group Feistel cipher as the public permutation:
\begin{align*}
\Psi_{k}^{f,g}(x)=\mathcal{F}_{g,f,f,g}(x\cdot k)\cdot k,
\end{align*}
where $k=(k^L,k^R)\in G^2$ is a key consisting of two subkeys, chosen independently and uniformly at random, and $f$ and $g$ are round functions on $G$, modelled as random function oracles, available to all parties, including the adversary. We consider the operation $x\cdot k$ for $x=(x^L,x^R)\in G^2$, to be the coordinate-wise group operation, but do not otherwise discern between it and the group operation $\cdot$ on elements of $G$. In the following, we shall follow the proof in \citep{GentryRamzan} closely. However, we make quite a few modifications, mostly due to the nature of our generalization. Note that we consider a one-key scheme, as opposed to the two-key version in \citep{GentryRamzan} (see Figure~\ref{GendGenRamPic}.) Our main theorem for this section is the following.

\begin{thm}\label{GentryRamzanMain}
Let $f,g$ be modelled as random oracles and let the subkeys of $k=(k^L,k^R)\in G^2$ be chosen independently and uniformly at random. Let $\Psi_k^{f,g}(x)=\mathcal{F}_{g,f,f,g}(x\cdot k)\cdot k$, and let $R\in_R \mathfrak{P}_{G^2 \rightarrow G^2}$. Then, for any probabilistic $4$-oracle adversary $\mathcal{A}$ with at most
\begin{itemize}
\item $q_c$ queries to $\Psi$ and $\Psi^{-1}$ (or $R$ and $R^{-1}$),
\item $q_f$ queries to $f$, and
\item $q_g$ queries to $g$,
\end{itemize}
we have
\begin{align*}
&\left| Pr\left[ \mathcal{A}^{\Psi,\Psi^{-1},f,g} = 1\right] - Pr\left[ \mathcal{A}^{R,R^{-1},f,g} = 1\right] \right| \\
	&\hspace*{20pt}\leq (2q_c^2 +4q_fq_c + 4q_gq_c + q_c^2 - q_c)|G|^{-1} + 2\cdot \begin{pmatrix}
	q_c \\ 2
	\end{pmatrix}(2|G|^{-1} + |G|^{-2}).
\end{align*}
\end{thm}

\subsection{Definitions}
Before we can begin the proof, we will need several definitions all of which are identical to the \citep{GentryRamzan} definitions, up to rewording.

\begin{defn}
Let $P$ denote the permutation oracle (either $\Psi$ or $R$), $\mathcal{O}^f$ and $\mathcal{O}^g$ the $f$ and $g$ oracles, respectively. We get the transcripts: $T_P$, the set of all $P$ queries, $T_f$, the set of all $f$ queries, and $T_g$, the set of all $g$ queries, i.e. the sets
\begin{align*}
T_P &= \lbrace \langle x_1,y_1 \rangle, \langle x_2,y_2 \rangle, \cdots, \langle x_{q_c},y_{q_c} \rangle \rbrace_P, \\
T_f &= \lbrace \langle x'_1,y'_1 \rangle, \langle x'_2,y'_2 \rangle, \cdots, \langle x'_{q_f},y'_{q_f} \rangle \rbrace_f, \\
T_g &= \lbrace \langle x''_1,y''_1 \rangle, \langle x''_2,y''_2 \rangle, \cdots, \langle x''_{q_g},y''_{q_g} \rangle \rbrace_g.
\end{align*}
We discern between two types of oracle queries: Cipher queries $(+,x)=P(x)$ and $(-,y)=P^{-1}(y)$; Oracle queries $(\mathcal{O}^f,x')$ and $(\mathcal{O}^g,x'')$, respectively $f$- and $g$-oracle queries.
\end{defn}

As we have no bounds on the computational complexity of the adversary $\mathcal{A}$, we may assume that $\mathcal{A}$ is deterministic, as we did in the proof of Theorem~\ref{PseudoEMbounded}. Hence, we may consider an algorithm $C_\mathcal{A}$ which, given a set of $\mathcal{A}$'s queries, can determine $\mathcal{A}$'s next query.

\begin{defn}
For $0\leq i \leq q_c$, $0\leq j \leq q_f$, and $0\leq k \leq q_g$, the $i+j+k+1$'st query by $\mathcal{A}$ is
\begin{align*}
C_\mathcal{A}\left[ \lbrace \langle x_1,y_1 \rangle, \ldots, \langle x_{i},y_{i} \rangle \rbrace_P, \lbrace \langle x'_1,y'_1 \rangle, \ldots, \langle x'_{j},y'_{j} \rangle \rbrace_f, \lbrace \langle x''_1,y''_1 \rangle, \ldots, \langle x''_{k},y''_{k} \rangle \rbrace_g \right]
\end{align*}
where the upper equality case on the indexes is defined to be $\mathcal{A}$'s final output.
\end{defn}

\begin{defn}
Let $\sigma = (T_P,T_f,T_g)$ be a tuple of transcripts with length $q_c,q_f,q_g$, respectively. We say that $\sigma$ is a \textbf{possible $\mathcal{A}$-transcript} if for every $1\leq i \leq q_c, 1 \leq j \leq q_f$, and $1\leq k \leq q_g$,
\begin{align*}
C_\mathcal{A}\left[ \lbrace \langle x_1,y_1 \rangle, \ldots, \langle x_{i},y_{i} \rangle \rbrace_P, \lbrace \langle x'_1,y'_1 \rangle, \ldots, \langle x'_{j},y'_{j} \rangle \rbrace_f, \lbrace \langle x''_1,y''_1 \rangle, \ldots, \langle x''_{k},y''_{k} \rangle \rbrace_g \right] \\ \hspace*{70pt}\in \lbrace (+,x_{i+1}), (-,y_{i+1}), (\mathcal{O}^f, x'_{j+1}), (\mathcal{O}^g,x''_{k+1})\rbrace .
\end{align*}
\end{defn}

Let us define two useful ways in which we may answer $\mathcal{A}$'s queries other than what we have already defined.

\begin{defn}
Let $\tilde{\Psi}$ be the process where the $\Psi$- and $\Psi^{-1}$ cipher query oracles use $f$ and $g$, and $\mathcal{O}^f$ uses $f$, but $\mathcal{O}^g$ is replaced by $\mathcal{O}^h$ for another, independent, random function $h$.
\end{defn}

\begin{defn}
Let $\tilde{R}$ denote the process which answers all \underline{oracle} queries using $f$ and $g$, but answers the $i$'th \underline{cipher} query as follows.
\begin{enumerate}
\item If $\mathcal{A}$ queries $(+,x_i)$ and there exists $1 \leq j < i$, such that the $j$'th query-answer pair has $x_j=x_i$, return $y_i:=y_j$.
\item If $\mathcal{A}$ queries $(-,y_i)$ and there exists $1 \leq j < i$, such that the $j$'th query-answer pair has $y_j=y_i$, return $x_i:=x_j$.
\item Otherwise, return uniformly chosen element in $G^2$.
\end{enumerate}
\end{defn}

The latter definition may not be consistent with any function or permutation, so we formalize exactly this event.

\begin{defn}\label{inconsistent}
Let $T_P$ be a possible $\mathcal{A}$-cipher-transcript. $T_P$ is \textbf{inconsistent} if for some $1\leq i < j \leq q_c$ there exist cipher-pairs such that either
\begin{itemize}
\item $x_i = x_j$ but $y_i \neq y_j$, or
\item $x_i \neq x_j$ but $y_i=y_j$.
\end{itemize}
Any $\sigma$ containing such a transcript $T_P$ is called \textbf{inconsistent}.
\end{defn}

\begin{note}
Assume from now on that $\mathcal{A}$ never repeats any part of a query if the answer can be determined from previous queries, i.e. every possible $\mathcal{A}$-transcript $\sigma$ is consistent such that if $i\neq j$, then $x_i \neq x_j$, $y_i\neq y_j$, $x'_i \neq x'_j$, and $x''_i \neq x''_j$.
\end{note}

\begin{note}
Let $T_\Psi, T_{\tilde{\Psi}}, T_{\tilde{R}}, T_R$ denote the transcripts seen by $\mathcal{A}$ when its cipher queries are answered by $\Psi, \tilde{\Psi}, \tilde{R}, R$, respectively, and oracle queries by $\mathcal{O}^f$ and $\mathcal{O}^g$ (noting that in the case of $\tilde{\Psi}$, the function in the $\mathcal{O}^g$ has been replaced by another random function, $h$.) We also note that using this notation, we have that $\mathcal{A}^{\Psi,\Psi^{-1},f,g} = C_\mathcal{A}(T_\Psi)$ (and likewise for $\tilde{\Psi}, \tilde{R}$, and $R$.)
\end{note}

\subsection{Lemmas}
Now, let us begin finding results that will aid us in proving our main theorem. First, we will compare the distributions of $\tilde{R}$ and $R$, using a result by Naor-Reingold\footnote{The proof of the proposition follows the argument of Proposition 3.3 in \citep{NaorReingold}.}. Afterwards, we shall consider when the distributions of $\Psi$ and $\tilde{\Psi}$ are equal. Lastly, we shall consider when the distributions of $\tilde{\Psi}$ and $\tilde{R}$ are equal. Combining these results will allow us to prove our main theorem.

We remark that whenever we write $k= (k^L,k^R)\in_R G^2$, we mean that the subkeys are chosen independently and uniformly at random.

\begin{lem}\label{Lem37}
$\left| \underset{\tilde{R}}{Pr}\left[ C_\mathcal{A}(T_{\tilde{R}})=1 \right] - \underset{R}{Pr}\left[ C_\mathcal{A}(T_{R})=1 \right] \right| \leq \begin{pmatrix}
q_c \\ 2
\end{pmatrix}\cdot |G|^{-2}$.
\end{lem}

\begin{proof}
Let $\sigma$ be a possible and consistent $\mathcal{A}$-transcript, then
\begin{align*}
\underset{R}{Pr}\left[ T_R = \sigma \right] = \begin{pmatrix}
|G|^{2} \\ q_c \end{pmatrix} = \underset{\tilde{R}}{Pr} \left[ T_{\tilde{R}} = \sigma \mid T_{\tilde{R}} \textit{ is consistent} \right],
\end{align*}
simply because the only difference between $T_R$ and $T_{\tilde{R}}$ is in the cipher queries, and when $T_{\tilde{R}}$ is consistent, we have no overlap on the query-answer pairs, hence we need only consider how to choose $q_c$ elements from $|G|^2$ many possible elements, without replacement. Let us now consider the probability of $T_{\tilde{R}}$ being inconsistent. If $T_{\tilde{R}}$ is inconsistent for some $1\leq i < j \leq q_c$ then either $x_i=x_j$ and $y_i \neq y_j$, or $x_i \neq x_j$ and $y_i = y_j$. For any given $i,j$, this happens with at most probability $|G|^{-2}$, because if $x_i=x_j$ is queried, then the $\tilde{R}$-oracle would return the corresponding $y_i=y_j$, but if $x_i \neq x_j$ is queried, then the $\tilde{R}$-oracle would return a uniformly random element (and likewise if $y_i=y_j$ or $y_i\neq y_j$ were queried to the inverse $\tilde{R}$-oracle.) Hence,
\begin{align*}
\underset{\tilde{R}}{Pr} \left[ T_{\tilde{R}} \textit{ is inconsistent} \right] \leq \begin{pmatrix}
q_c \\ 2 \end{pmatrix} \cdot |G|^{-2}.
\end{align*}
We thereby get that,
\begin{footnotesize}
\begin{align*}
&\left| \underset{\tilde{R}}{Pr}\left[ C_\mathcal{A}(T_{\tilde{R}})=1 \right] - \underset{R}{Pr}\left[ C_\mathcal{A}(T_{R})=1 \right] \right| \\
&\leq \left| \underset{\tilde{R}}{Pr}\left[ C_\mathcal{A}(T_{\tilde{R}})=1 | T_{\tilde{R}} \textit{ is consistent} \right] - \underset{R}{Pr}\left[ C_\mathcal{A}(T_{R})=1 \right]\right| \cdot \underset{\tilde{R}}{Pr}\left[ T_{\tilde{R}} \textit{ is consistent} \right] \\
	&\hspace*{10pt} + \left| \underset{\tilde{R}}{Pr}\left[ C_\mathcal{A}(T_{\tilde{R}})=1 | T_{\tilde{R}} \textit{ is inconsistent} \right] - \underset{R}{Pr}\left[ C_\mathcal{A}(T_{R})=1 \right]\right| \cdot \underset{\tilde{R}}{Pr}\left[ T_{\tilde{R}} \textit{ is inconsistent} \right] \\
&\leq \underset{\tilde{R}}{Pr}\left[ T_{\tilde{R}} \textit{ is inconsistent} \right] \\
&\leq \begin{pmatrix} q_c \\ 2 \end{pmatrix} \cdot |G|^{-2},
\end{align*}
\end{footnotesize}
as the distribution over $R$ is independent of the (in)consistency of $T_{\tilde{R}}$.
\end{proof}

Let us now focus on the distributions of $T_\Psi$ and $T_{\tilde{\Psi}}$, to show that they are identical unless the input to $g$ in the cipher query to $\Psi$ is equal to the oracle input to $h$ in $\mathcal{O}^h$. In order to do so, we first define the event $\textsf{BadG}(k)$.

\begin{defn}
For every specific key $k=(k^L,k^R)\in_R G^2$, we define $\textsf{BadG}(k)$ to be the set of all possible and consistent $\mathcal{A}$-transcripts $\sigma$, satisfying at least one of the following:
\begin{itemize}
\item[\textbf{BG1}:] $\exists i,j, 1\leq i \leq q_c, 1 \leq j \leq q_g$, such that $x_i^R \cdot k^R = x''_j$, or
\item[\textbf{BG2}:] $\exists i,j, 1\leq i \leq q_c, 1 \leq j \leq q_g$, such that $y_i^L \cdot (k^L)^{-1} = x''_j$.
\end{itemize}
\end{defn}

\begin{lem}\label{Lem39}
Let $k=(k^L,k^R)\in_R G^2$. For any possible and consistent $\mathcal{A}$-transcript $\sigma=(T_P,T_f,T_g)$, we have
\begin{align*}
\underset{k}{Pr}\left[ \sigma \in \textsf{BadG}(k) \right] \leq \frac{2q_gq_c}{|G|}.
\end{align*}
\end{lem}

\begin{proof}
We know that $\sigma \in \textsf{BadG}(k)$ if one of \textbf{BG1} or \textbf{BG2} occur, hence, using the union bound,
\begin{align*}
\underset{k}{Pr}\left[ \sigma \in \textsf{BadG}(k) \right] &= \underset{k}{Pr}\left[ \textbf{BG1} \text{ occurs } \vee \textbf{BG2} \text{ occurs } | \sigma \right] \\
	&\leq  \underset{k}{Pr}\left[ \textbf{BG1} \text{ occurs } | \sigma \right] +  \underset{k}{Pr}\left[ \textbf{BG2} \text{ occurs } | \sigma \right] \\
	&\leq q_gq_c\cdot |G|^{-1} + q_gq_c\cdot |G|^{-1} \\
	&= 2q_gq_c\cdot |G|^{-1}.
\end{align*}
\end{proof}

\begin{lem}\label{NotBadPsitoBarPsi}
Let $\sigma$ be a possible and consistent $\mathcal{A}$-transcript, then
\begin{align*}
\underset{\Psi}{Pr}\left[ T_{\Psi} = \sigma | \sigma \not\in \textsf{BadG}(k) \right] = \underset{\tilde{\Psi}}{Pr}\left[ T_{\tilde{\Psi}} = \sigma \right].
\end{align*}
\end{lem}

\begin{proof}
We want to show that the query answers in the subtranscripts of the games $\Psi$ and $\tilde{\Psi}$ are equally distributed, under the condition that neither of the events \textbf{BG1} nor \textbf{BG2} occur in game $\Psi$. Fix the key $k=(k^L,k^R)\in_R G^2$. Recall that the adversary does not query an oracle if it can determine the answer from previous queries.

In both games, for any $\mathcal{O}^f$-oracle query $x' \in G$, the query answer will be equally distributed in both games as the underlying random function $f$ is the same in both games.

In game $\Psi$, an $\mathcal{O}^g$-oracle query, $x''\in G$, will have a uniformly random answer as $g$ is a random function. Likewise, in game $\tilde{\Psi}$, an $\mathcal{O}^g$-oracle query, $x'' \in G$, will have a uniformly random answer as $h$ is a random function.

Consider now the permutation oracle $P= \mathcal{F}_{g,f,f,g}(x\cdot k)\cdot k$. We consider a query-answer pair $\langle x, y \rangle \in T_P$ for $x,y\in G^2$.

In both games, $x^R\cdot k^R$ will be the input to the first round function, which is $g$. In game $\tilde{\Psi}$ the output is always a uniformly random element, newly selected by $g$. In game $\Psi$, if $x^R\cdot k^R$ has already been queried to the $\mathcal{O}^g$-oracle, the output of the round function is the corresponding oracle answer, else it is a uniformly random element, newly selected by $g$.  As the former event in game $\Psi$ never occurs because the event \textbf{BG1} never occurs, the distributions are equal.

As both games have access to the same random function $f$, the second and third round function outputs will have equal distributions.

In both games, $y^L\cdot (k^L)^{-1}$ will be the input to the fourth round function, which is again $g$. In game $\tilde{\Psi}$ the output is always a uniformly random element, newly selected by $g$, unless $y^L\cdot (k^L)^{-1} = x^R\cdot k^R$, in which case the output is equal to the output of the first round function. In game $\Psi$, if $x^R\cdot k^R$ has already been queried to the $\mathcal{O}^g$-oracle, but not as input to the first round function, the output of the round function is the corresponding oracle answer. If $y^L\cdot (k^L)^{-1} = x^R\cdot k^R$, then the output is equal to the output of the first round function, else it is a uniformly random element newly selected by $g$. As the former event in game $\Psi$ never occurs because the event \textbf{BG2} never occurs, the distributions are equal.

As $\mathcal{A}$ does not ask a query if it can determine the answer based on previous queries, we see that the inverse permutation oracle, using $P^{-1}$, yields analogous distributions. Thus, the distributions for the two games must be equal.
\end{proof}

Let us show that the distributions of $T_{\tilde{\Psi}}$ and $T_{\tilde{R}}$ are identical, unless the same value is input to $f$ on two separate occasions. Here we also define when a key is "bad" as we did above, but altered such that it pertains to our current oracles.

\begin{defn}
For every specific key $k=(k^L,k^R)\in_R G^2$ and function $g\in_R \mathfrak{F}_{G\rightarrow G}$, define $\textsf{Bad}(k,g)$ to be the set of all possible and consistent $\mathcal{A}$-transcripts $\sigma$ satisfying at least one of the following events:
\begin{itemize}
\item[\textbf{B1}:] $\exists 1\leq i < j \leq q_c$, such that
\begin{align*}
x_i^L\cdot k^L \cdot g(x_i^R\cdot k^R) = x_j^L \cdot k^L \cdot g(x_j^R \cdot k^R)
\end{align*}
\item[\textbf{B2}:] $\exists 1\leq i < j \leq q_c$, such that
\begin{align*}
y_i^R\cdot (k^R)^{-1} \cdot \left(g(y_i^L\cdot (k^L)^{-1})\right)^{-1} = y_j^R\cdot (k^R)^{-1} \cdot \left(g(y_j^L\cdot (k^L)^{-1})\right)^{-1}
\end{align*}
\item[\textbf{B3}:] $\exists 1\leq i , j \leq q_c$, such that
\begin{align*}
x_i^L\cdot k^L \cdot g(x_i^R\cdot k^R) = y_j^R\cdot (k^R)^{-1} \cdot \left(g(y_j^L\cdot (k^L)^{-1})\right)^{-1}
\end{align*}
\item[\textbf{B4}:] $\exists 1\leq i \leq q_c, 1\leq j \leq q_f$, such that
\begin{align*}
x_i^L\cdot k^L \cdot g(x_i^R\cdot k^R) = x'_j
\end{align*}
\item[\textbf{B5}:] $\exists 1\leq i \leq q_c, 1\leq j \leq q_f$, such that
\begin{align*}
y_i^R\cdot (k^R)^{-1} \cdot \left(g(y_i^L\cdot (k^L)^{-1})\right)^{-1} = x'_j
\end{align*}
\end{itemize}
\end{defn}

\begin{lem}\label{Lem42}
Let $k=(k^L,k^R)\in_R G^2$. For any possible and consistent $\mathcal{A}$-transcript $\sigma$, we have that
\begin{align*}
\underset{k,g}{Pr}\left[ \sigma \in \textsf{Bad}(k,g)\right] \leq \left( q_c^2+2q_fq_c + 2\cdot \begin{pmatrix} q_c \\ 2 \end{pmatrix} \right)\cdot |G|^{-1}.
\end{align*}
\end{lem}

\begin{proof}
We have that $\sigma \in \textsf{Bad}(k,g)$ if it satisfies a $\textit{\textbf{Bi}}$ for some $\textbf{i}=\lbrace1,\ldots,5\rbrace$. Using that $k^L,k^R$ are uniform and independently chosen, and $g\in_R \mathfrak{F}_{G\rightarrow G}$, we may achieve an upper bound on the individual event probabilities, and then use the union bound.

There are $\begin{pmatrix} q_c \\ 2 \end{pmatrix}$ many ways of picking $i,j$ such that $1\leq i< j \leq q_c$, also, $q_fq_c$ many ways of picking $i,j$ such that $1\leq i \leq q_c, 1 \leq j \leq q_f$, and $q_c^2$ many ways of picking $i,j$ such that $1 \leq i,j\leq q_c$. The probability that two elements chosen from $G$ are equal is $|G|^{-1}$, so we may bound each event accordingly and achieve, using the union bound, that
\begin{footnotesize}
\begin{align*}
\underset{k,g}{Pr}\left[ \sigma \in \textsf{Bad}(k,g)\right] &= \underset{k,g}{Pr}\left[ \bigvee_{i=1}^5 \textit{\textbf{Bi}} \text{ occurs } | \sigma \right] \\
&\leq \sum_{i=1}^5 \underset{k,g}{Pr}\left[ \textit{\textbf{Bi}} \text{ occurs } | \sigma \right] \\
&\leq \begin{pmatrix} q_c \\ 2 \end{pmatrix}\cdot |G|^{-1} + \begin{pmatrix} q_c \\ 2 \end{pmatrix}\cdot |G|^{-1} + q_c^2 \cdot |G|^{-1} + q_fq_c \cdot |G|^{-1} + q_fq_c \cdot |G|^{-1} \\
&= \left( q_c^2 + 2q_fq_c + 2\begin{pmatrix} q_c \\ 2 \end{pmatrix} \right) \cdot |G|^{-1}.
\end{align*}
\end{footnotesize}
\end{proof}

\begin{lem}\label{Lem43}
Let $\sigma$ be a possible and consistent $\mathcal{A}$-transcript, then
\begin{align*}
\underset{\tilde{\Psi}}{Pr}\left[ T_{\tilde{\Psi}} = \sigma | \sigma \not\in \textsf{Bad}(k,g)\right] = \underset{\tilde{R}}{Pr}\left[ T_{\tilde{R}} = \sigma \right].
\end{align*}
\end{lem}

The following proof is based on the proof in \citep{GentryRamzan} which refers to \citep{NaorReingold} for the first part of their argument. We need the generalization of this argument and so also include it.

\begin{proof}
Since $\sigma$ is a possible $\mathcal{A}$-transcript, we have for all $1\leq i \leq q_c, 1\leq j \leq q_f, 1\leq k \leq q_g$:
\begin{align*}
C_\mathcal{A}\left[ \lbrace \langle x_1,y_1 \rangle, \ldots, \langle x_{i},y_{i} \rangle \rbrace_P, \lbrace \langle x'_1,y'_1 \rangle, \ldots, \langle x'_{j},y'_{j} \rangle \rbrace_f, \lbrace \langle x''_1,y''_1 \rangle, \ldots, \langle x''_{k},y''_{k} \rangle \rbrace_g \right] \\ \hspace*{70pt}\in \lbrace (+,x_{i+1}), (-,y_{i+1}), (\mathcal{O}^f, x'_{j+1}), (\mathcal{O}^g,x''_{k+1})\rbrace .
\end{align*}
Therefore, $T_{\tilde{R}}=\sigma$ if and only if $\forall 1\leq i \leq q_c, \forall 1\leq j \leq q_f$, and $\forall 1\leq k \leq q_g$, the $i,j,k$'th respective answers $\tilde{R}$ gives are $y_i$ or $x_i$, and $x'_j$ and $x''_k$, respectively. As $\mathcal{A}$ never repeats any part of a query, we have, by the definition of $\tilde{R}$, that the $i$'th cipher-query answer is an independent and uniform element of $G^2$, and as $f$ and $g$ were modelled as random function oracles, so too will their oracle outputs be independent and uniform elements of $G$. Hence,
\begin{align*}
\underset{\tilde{R}}{Pr}\left[ T_{\tilde{R}} = \sigma \right] = |G|^{-(2q_c+q_f+q_g)}.
\end{align*}

For the second part of this proof, we fix $k,g$ such that $\sigma \not\in \textsf{Bad}(k,g)$ and seek to compute $\underset{f,h}{Pr}\left[ T_{\tilde{\Psi}} = \sigma \right]$.
Since $\sigma$ is a possible $\mathcal{A}$-transcript, we have that $T_{\tilde{\Psi}}= \sigma$ if and only if
\begin{itemize}
\item $y_i = \mathcal{F}_{g,f,f,g}(x_i\cdot k)\cdot k$ for all $1\leq i \leq q_c$,
\item $y'_j = f(x'_j)$ for all $1\leq j \leq q_f$, and
\item $y''_k = g(x''_k)$ for all $1\leq k \leq q_g$ (note that $g=h$ here.)
\end{itemize}
If we define
\begin{align*}
X_i &:= x_i^L\cdot k^L \cdot g(x_i^R\cdot k^R) \\
Y_i &:= y_i^R \cdot (k^R)^{-1} \cdot \left(g(y_i^L\cdot (k^L)^{-1})\right)^{-1},
\end{align*}
then $(y_i^L, y_i^R) = \tilde{\Psi}(x_i^L,x_i^R)$  if and only if
\begin{center}
$k^R\cdot f(X_i) = (x_i^R)^{-1} \cdot Y_i$ \hspace*{5pt}  and \hspace*{5pt}  $X_i\cdot f(Y_i) = y_i^L \cdot (k^L)^{-1}$,
\end{center}
where the second equality of the latter is equivalent to $(k^L)^{-1}\cdot (f(Y_i))^{-1} = (y_i^L)^{-1}\cdot X_i$. Observe that, for all $1 \leq i < j \leq q_c$, $X_i \neq X_j$ (by \textit{\textbf{B1}}) and $Y_i \neq Y_j$ (by \textit{\textbf{B2}}.) Similarly, $1 \leq i < j \leq q_c$, $X_i \neq Y_j$ (by \textit{\textbf{B3}}.) Also, for all $1 \leq i \leq q_c$ and for all $1 \leq j \leq q_f$, $x'_j \neq X_i$ (by \textit{\textbf{B4}}) and $x'_j \neq Y_i$ (by \textit{\textbf{B5}}.) Hence, $\sigma \not\in \textsf{Bad}(k,g)$ implies that all inputs to $f$ are distinct. This then implies that $Pr_{f,h}\left[ T_{\tilde{\Psi}} = \sigma \right] = |G|^{-(2q_c+q_f+q_g)}$ as $h$ was also modelled as a random function, independent from $g$.
Thus, as we assumed that $k$ and $g$ were chosen such that $\sigma \not\in \textsf{Bad}(k,g)$,
\begin{align*}
\underset{\tilde{\Psi}}{Pr}\left[ T_{\tilde{\Psi}} = \sigma | \sigma \not\in \textsf{Bad}(k,g)\right] &= |G|^{-(2q_c+q_f+q_g)} =  \underset{\tilde{R}}{Pr}\left[ T_{\tilde{R}} = \sigma \right].
\end{align*}
\end{proof}

\subsection{Proof of Theorem~\ref{GentryRamzanMain}}
To complete the proof of Theorem~\ref{GentryRamzanMain}, we combine the above lemmas into the following probability estimation.

\begin{proof}[Proof of Theorem~\ref{GentryRamzanMain}]
Let $\Gamma$ be the set of all possible and consistent $\mathcal{A}$-transcripts $\sigma$ such that $\mathcal{A}(\sigma)=1$. In the following, we ease notation, for the sake of the reader. We let $\textsf{BadG}(k)$ be denoted by $BadG$ and $\textsf{Bad}(k,g)$ by $Bad$. Furthermore, we abbreviate inconsistency as $incon.$. Let us consider the cases between $\Psi,\tilde{\Psi}$ and $\tilde{R}$.

\begin{footnotesize}
\begin{align*}
&\left| Pr_{\Psi}\left[ C_\mathcal{A}(T_\Psi)=1\right] - Pr_{\tilde{\Psi}}\left[ C_\mathcal{A}(T_{\tilde{\Psi}})=1\right] \right| \\
&\leq \left| \sum_{\sigma\in\Gamma} \left( Pr_{\Psi}\left[ T_\Psi = \sigma \right] - Pr_{\tilde{\Psi}}\left[ T_{\tilde{\Psi}} = \sigma \right]\right) \right| + Pr_{\tilde{\Psi}}\left[ T_{\tilde{\Psi}} \hspace*{4pt} incon.\right] \\
&\leq \sum_{\sigma\in \Gamma} \left| Pr_\Psi \left[ T_{\Psi} = \sigma \mid \sigma \not\in BadG \right] - Pr_{\tilde{\Psi}}\left[ T_{\tilde{\Psi}} = \sigma \right] \right| \cdot Pr_k \left[ \sigma \not\in BadG \right] \\
	&\hspace*{15pt}+ \left| \sum_{\sigma \in \Gamma} \left( Pr_{\Psi} \left[ T_\Psi = \sigma \mid \sigma \in BadG \right] - Pr_{\tilde{\Psi}} \left[ T_{\tilde{\Psi}} = \sigma \right]\right) \cdot Pr_k \left[ \sigma \in BadG \right]\right| \\
		&\hspace*{30pt}+ Pr_{\tilde{\Psi}}\left[ T_{\tilde{\Psi}} \hspace*{4pt} incon. \right] \\
&\leq \left| \sum_{\sigma \in \Gamma} \left( Pr_{\Psi} \left[ T_\Psi = \sigma \mid \sigma \in BadG \right] - Pr_{\tilde{\Psi}} \left[ T_{\tilde{\Psi}} = \sigma \right]\right) \cdot Pr_k \left[ \sigma \in BadG \right]\right| + q_c(q_c-1)|G|^{-1},
\end{align*}
\end{footnotesize}
where we in the last estimate used Lemma~\ref{NotBadPsitoBarPsi} and a consideration of the maximal amount of possible inconsistent pairs.

At the same time,
\begin{footnotesize}
\begin{align*}
&\left| Pr_{\tilde{\Psi}}\left[ C_\mathcal{A}(T_{\tilde{\Psi}})=1\right] - Pr_{\tilde{R}}\left[ C_\mathcal{A}(T_{\tilde{R}})=1\right] \right| \\
&\leq \left| \sum_{\sigma\in\Gamma} \left( Pr_{\tilde{\Psi}}\left[ T_{\tilde{\Psi}} = \sigma \right] - Pr_{\tilde{R}}\left[ T_{\tilde{R}} = \sigma \right]\right) \right| + Pr_{\tilde{R}}\left[ T_{\tilde{R}} \hspace*{4pt} incon.\right] + Pr_{\tilde{\Psi}}\left[ T_{\tilde{\Psi}} incon. \right] \\
&\leq \sum_{\sigma\in \Gamma} \left| Pr_{\tilde{\Psi}} \left[ T_{\tilde{\Psi}} = \sigma \mid \sigma \not\in Bad \right] - Pr_{\tilde{R}}\left[ T_{\tilde{R}} = \sigma \right] \right| \cdot Pr_k \left[ \sigma \not\in Bad \right] \\
	&\hspace*{15pt} + \left| \sum_{\sigma \in \Gamma} \left( Pr_{\tilde{\Psi}} \left[ T_{\tilde{\Psi}} = \sigma \mid \sigma \in Bad \right] - Pr_{\tilde{R}} \left[ T_{\tilde{R}} = \sigma \right]\right) \cdot Pr_k \left[ \sigma \in Bad \right]\right| \\
		&\hspace*{30pt}+ Pr_{\tilde{R}}\left[ T_{\tilde{R}} \hspace*{4pt} incon. \right] + Pr_{\tilde{\Psi}}\left[ T_{\tilde{\Psi}} incon. \right] \\
&\leq \left| \sum_{\sigma \in \Gamma} \left( Pr_{\tilde{\Psi}} \left[ T_{\tilde{\Psi}} = \sigma \mid \sigma \in Bad \right] - Pr_{\tilde{R}} \left[ T_{\tilde{R}} = \sigma \right]\right) \cdot Pr_k \left[ \sigma \in Bad \right]\right| + \begin{pmatrix} q_c \\ 2 \end{pmatrix}|G|^{-2} + 2\begin{pmatrix} q_c \\ 2 \end{pmatrix}|G|^{-1},
\end{align*}
\end{footnotesize}
where we in the last estimate used Lemma~\ref{Lem43} and the proof of Lemma~\ref{Lem37}.

Let us use the above in a temporary estimate,
\begin{footnotesize}
\begin{align}
&\left| Pr_\Psi\left[ C_\mathcal{A}(T_\Psi)=1\right] - Pr_R\left[C_\mathcal{A}(T_R)=1\right] \right| \nonumber \\
&= \left| Pr_\Psi\left[ C_\mathcal{A}(T_\Psi)=1\right] - Pr_{\tilde{\Psi}}\left[C_\mathcal{A}(T_{\tilde{\Psi}})=1\right] \right| \nonumber \\
	&\hspace*{40pt}+ \left| Pr_{\tilde{\Psi}}\left[ C_\mathcal{A}(T_{\tilde{\Psi}})=1\right] - Pr_{\tilde{R}}\left[C_\mathcal{A}(T_{\tilde{R}})=1\right] \right| \nonumber \\
		&\hspace*{80pt}+ \left| Pr_{\tilde{R}}\left[ C_\mathcal{A}(T_{\tilde{R}})=1\right] - Pr_R\left[C_\mathcal{A}(T_R)=1\right] \right| \nonumber \\
&\leq \left| \sum_{\sigma \in \Gamma} \left( Pr_{\Psi} \left[ T_\Psi = \sigma \mid \sigma \in BadG \right] - Pr_{\tilde{\Psi}} \left[ T_{\tilde{\Psi}} = \sigma \right]\right) \cdot Pr_k \left[ \sigma \in BadG \right]\right| + q_c(q_c-1)|G|^{-1} \nonumber\\
	&\hspace*{30pt}+ \left| \sum_{\sigma \in \Gamma} \left( Pr_{\tilde{\Psi}} \left[ T_{\tilde{\Psi}} = \sigma \mid \sigma \in Bad \right] - Pr_{\tilde{R}} \left[ T_{\tilde{R}} = \sigma \right]\right) \cdot Pr_k \left[ \sigma \in Bad \right]\right| + \begin{pmatrix} q_c \\ 2 \end{pmatrix}|G|^{-2} + 2\begin{pmatrix} q_c \\ 2 \end{pmatrix}|G|^{-1} \nonumber\\
		&\hspace*{60pt}+ \begin{pmatrix} q_c \\ 2 \end{pmatrix}|G|^{-2} \label{GenRamMainEstimate},
\end{align}
\end{footnotesize}
where we in the last estimate also used Lemma~\ref{Lem37}.

We may assume WLOG that
\begin{footnotesize}
\begin{align*}
\sum_{\sigma \in \Gamma} Pr_\Psi\left[ T_\Psi = \sigma \mid \sigma \in BadG \right] \cdot Pr_k\left[ \sigma \in BadG \right] \leq \sum_{\sigma \in \Gamma} Pr_{\tilde{\Psi}}\left[ T_{\tilde{\Psi}} = \sigma\right] \cdot Pr_k\left[ \sigma \in BadG \right]
\end{align*}
\end{footnotesize}
and likewise,
\begin{footnotesize}
\begin{align*}
\sum_{\sigma \in \Gamma} Pr_{\tilde{\Psi}}\left[ T_{\tilde{\Psi}} = \sigma \mid \sigma \in BadG \right] \cdot Pr_k\left[ \sigma \in BadG \right] \leq \sum_{\sigma \in \Gamma} Pr_{\tilde{R}}\left[ T_{\tilde{R}} = \sigma\right] \cdot Pr_k\left[ \sigma \in BadG \right],
\end{align*}
\end{footnotesize}
such that by Lemma~\ref{Lem39}, respectively Lemma~\ref{Lem42}, we get the following continued estimate from (\ref{GenRamMainEstimate}), using the triangle inequality and that $|\Gamma|\leq |G|^{2q_c+q_f+q_g}$ (every combination of query elements).
\begin{footnotesize}
\begin{align*}
&\left| Pr_\Psi\left[ C_\mathcal{A}(T_\Psi)=1\right] - Pr_R\left[C_\mathcal{A}(T_R)=1\right] \right| \\
&\leq 2 \sum_{\sigma \in \Gamma} Pr_{\tilde{\Psi}} \left[T_{\tilde{\Psi}} = \sigma \right]\cdot Pr_k \left[ \sigma \in BadG \right] + 2q_c(q_c-1)|G|^{-1} \\
	&\hspace*{30pt}+ 2\sum_{\sigma \in \Gamma} Pr_{\tilde{R}}\left[ T_{\tilde{R}}= \sigma\right] \cdot Pr_k \left[\sigma \in Bad \right] \\
		&\hspace*{60pt}+ 2\begin{pmatrix} q_c \\ 2 \end{pmatrix}|G|^{-2} \\
&\leq 2|\Gamma|\cdot |G|^{-(2q_c+q_f+q_g)}\cdot \max_{\sigma \in \Gamma} Pr_k \left[ \sigma \in BadG \right] + 2q_c(q_c-1)|G|^{-1} \\
	&\hspace*{30pt}+ 2|\Gamma|\cdot |G|^{-(2q_c+q_f+q_g)}\cdot\max_{\sigma\in\Gamma} Pr_k \left[ \sigma \in Bad \right] \\
		&\hspace*{60pt}+ 2\begin{pmatrix} q_c \\ 2 \end{pmatrix}|G|^{-2} \\
&\leq 4q_gq_c\cdot |G|^{-1} + 2q_c(q_c-1)|G|^{-1} + 2\left(q_c^2 + 2q_fq_c + 2\begin{pmatrix} q_c \\ 2 \end{pmatrix}\right)|G|^{-1} + 2\begin{pmatrix} q_c \\ 2 \end{pmatrix}|G|^{-2} \\
&= (2q_c^2+4q_gq_c + 4q_fq_c + 2q_c^2-2q_c)|G|^{-1} + 2\begin{pmatrix} q_c \\ 2 \end{pmatrix}\left(2|G|^{-1} + |G|^{-2}\right).
\end{align*}
\end{footnotesize}
\end{proof}

If we denote the total amount of queries as $q=q_c+q_f+q_g$, then we may quickly estimate and reword the main theorem as:

\begin{thm}
Let $f,g$ be modelled as random oracles, let $k=(k^L,k^R)\in_R G^2$, let $\Psi_k^{f,g}(x)=\mathcal{F}_{g,f,f,g}(x\cdot k)\cdot k$, and let $R\in_R \mathfrak{P}_{G^2 \rightarrow G^2}$. Then, for any $4$-oracle adversary $\mathcal{A}$, with at most $q$ total queries, we have
\begin{align*}
\left| Pr\left[ \mathcal{A}^{\Psi,\Psi^{-1},f,g} = 1\right] - Pr\left[ \mathcal{A}^{R,R^{-1},f,g} = 1\right] \right| \leq 2(3q^2-2q)|G|^{-1} + (q^2-q)|G|^{-2}.
\end{align*}
\end{thm}

\begin{proof}
Given Theorem~\ref{GentryRamzanMain}, we get, by using that $q_f,q_g\geq 0$,
\begin{align*}
&q_c^2 +2q_fq_c + 2q_gq_c + q_c^2 - q_c \\ 
	&= 2(q_c^2 + q_fq_c + q_gq_c) - q_c \\
	&\leq 2(q_c^2 + q_fq_c + q_gq_c) + (2(q_f+q_g)^2 + 2(q_fq_c+q_gq_c) -q_f-q_g)-q_c \\
	&= 2(q_c^2 + 2q_fq_c  + q_f^2 + 2q_fq_g + 2q_gq_c+ q_g^2) - (q_c+q_f+q_g) \\
	&= 2(q_c+q_f+q_g)^2 - (q_c+q_f+q_g) \\
	&= 2q^2-q.
\end{align*}
As $2\cdot\begin{pmatrix} q_c \\ 2 \end{pmatrix} =q_c^2-q_c \leq q^2-q$, we get the final estimate by some reordering.
\end{proof}

\newpage
\section{Quantum Pseudorandom Permutations}\label{ZhandryResultSection}
A major result in Quantum Cryptography is Zhandry's result \citep{Zhandry} on the implication of the existence of Quantum One-Way Functions to the existence of Quantum Pseudorandom Permutations . His work escapes the use of the existence of One-Way Functions (OWFs) implying the existence of Pseudorandom Generators, implying the existence of Pseudorandom Functions , which in turn imply the existence of Pseudorandom Permutations (PRPs), as we know from the classical setting. Instead, Zhandry analyzes PRP constructions and defines an information theoretic object called a Function to Permutation Converter (FPC) which he claims lies implicit in many constructions of PRPs (especially those from OWFs.) The well-written article refers to many classical results to get its main result, however, one lemma in particular refers to a collection of works based around card shuffles. In the following, we attempt to reproduce Zhandry's work, but over an arbitrary group, and in doing so, aim to generalize the card shuffle results as well.

\subsection{Generalizing Zhandry's Main Result}
Let us first consider what it means to be a Quantum Pseudorandom Permutation. In the classical setting, we consider an adversary employing an efficient computer, i.e. it is a polynomial-time algorithm. The classical adversary interacts with the oracle, sending and receiving bit strings. In the quantum setting, we have an adversary who can use a quantum computer, which does not have polynomial running time. The oracle interactions may be modelled in two ways: using only classical queries (e.g. as we modelled in the first part of the thesis) or using quantum superposition queries. It is not entirely important for this thesis what the quantum superposition queries are but it should be noted that using such superposition queries, an adversary may in theory query the entire message space using only polynomially many queries, giving us a very strong security model. In Zhandry's case, he shows security against the latter adversary: a quantum computer asking quantum superposition queries.

Intuitively, a Function to Permutation Converter (FPC) is an algorithm $P$ on a domain and image space $X$, making oracle queries to a function $O$, such that $P^O$, i.e. the algorithm $P$ making calls to the oracle $O$, is a permutation on $X$. It must also hold that $P^O$ is indistinguishable from a uniformly random permutation, whenever $O$ is a uniformly random function. Note that this intuitive definition works just as well if we consider $X=G$ to be an arbitrary group.

Using a classical FPC which is secure on the entire domain $X$, i.e. when the adversary makes $|X|$ queries, we may build a quantum-secure FPC by the following reduction: Assume there exists a quantum FPC adversary which the classical FPC adversary runs as a subroutine. As the classical FPC adversary may query the entire domain, it can answer any query made by the quantum FPC adversary and the reduction may be completed. If we assume that the oracle function employed by the quantum-secure FPC is a quantum-secure PRF, then we immediately get a quantum-secure PRP.

\subsubsection{Preliminaries}
We wish to use Zhandry's notation, for its convenience, and therefore generalize some notions and introduce others.

Recall that we consider the positive integer $\lambda$ to be the security parameter, specified in unary per convention. We assume that for each $\lambda$ there exists a uniquely specified group $G(\lambda) = G_\lambda \in \mathcal{G}$ with size $|G_\lambda| \geq 2^\lambda$.

\begin{defn}
A \textbf{\textsf{CComp-} (resp. \textsf{QComp-}) one-way function} is an efficient classical function $\text{OWF}: G_\lambda \rightarrow G_*$ such that, for any efficient classical (resp. quantum) adversary $\mathcal{A}$, the probability that $\mathcal{A}$ inverts $\text{OWF}$ on a random input is negligible. That is, there exists a negligible function $negl(\cdot)$ such that
\begin{align*}
\underset{x \in_R G_\lambda}{Pr}\left[ OWF(\mathcal{A}(\lambda, \text{OWF}(x)))= \text{OWF}(x) \right] < negl(\lambda).
\end{align*}
\end{defn}

\begin{defn}
A \textbf{\textsf{(C,Q)-}pseudorandom function (PRF)}, for a pair $\textsf{(C,Q)}\in \lbrace \textsf{(CComp,CQuery)},\textsf{(QComp,CQuery)},\textsf{(QComp,QQuery)}\rbrace$ is a family of efficient classical functions $\text{PRF}_{m,n}:G_\lambda \times G_m \rightarrow G_n$ such that the following holds. For any efficient (resp. quantum) adversary $\mathcal{A}$, $\mathcal{A}$ cannot distinguish $\text{PRF}_{m,n}(k,\cdot)$, for a random $k\in_R G_\lambda$, from a uniformly random function $O\in_R \mathfrak{F}_{G_m\rightarrow G_n}$, the set of functions from $G_m$ to $G_n$. That is, there exists a negligible function $negl(\cdot)$ such that
\begin{align*}
\left| \underset{k\in_R G_{\lambda}}{Pr}\left[ \mathcal{A}^{\text{PRF}_{m,n}(k,\cdot)}(\lambda) = 1 \right] - Pr\left[ \mathcal{A}^{O(\cdot)}(\lambda) = 1 \right] \right| < negl(\lambda).
\end{align*}
\end{defn}

If $PRF: G \times G \rightarrow G$, for some group $G\in \mathcal{G}$, we say that it is a \textbf{pseudorandom function on $G$}.

Let us explain some of the notation used in the above definitions. $\textsf{(C,Q)}$ represents the computational capability of the adversary and the type of queries it is allowed. Hence, $\textsf{(CComp,CQuery)}$ implies classical computation with classical queries, $\textsf{(QComp,CQuery)}$ implies quantum computation with classical queries, and $\textsf{(QComp,QQuery)}$ implies quantum computation with quantum queries. Note that we do not consider classical adversaries making quantum queries.

We make one, rather large, assumption in this part of the thesis, which is that we assume the following theorem holds over groups $G\in \mathcal{G}$.

\begin{thm}\label{OWimplyPRF}
If \textsf{CComp}-one-way functions exist, then so do \textsf{(CComp,CQuery)}-PRFs. Moreover, if \textsf{QComp}-one-way functions exist, then so do \textsf{(QComp,QQuery)}-PRFs.
\end{thm}

The above theorem, in the bit strings case, follows by combining the results of \citep{GGM86, HILL99, 2Zhandry}.

We also redefine our notion of a Pseudorandom Permutation, as this will be crucial to the remainder of Zhandry's work, btu does not impact our prior work.

\begin{defn}
A \textbf{\textsf{(C,Q,D)}-pseudorandom permutation (PRP)}, for an ordered pair $\textsf{(C,Q)}\in \lbrace \textsf{(CComp,CQuery)},\textsf{(QComp,CQuery)},\textsf{(QComp,QQuery)}\rbrace$ and $\textsf{D}\in \lbrace \textsf{LargeD, AnyD}\rbrace$ is a family of efficient classical function pairs $\textsf{PRP}_o: G_\lambda \times G_o \rightarrow G_o$ and $\textsf{PRP}_o^{-1}: G_\lambda \times G_o \rightarrow G_o$ such that the following holds.

First, for every key $k$ and integer $o$, the functions $\textsf{PRP}_o(k,\cdot)$ and $\textsf{PRP}_o^{-1}(k,\cdot)$ are inverses of each other. That is, $\textsf{PRP}_o^{-1}(k,\textsf{PRP}_o(k,x))=x$ for all $k,x$. This implies that $\textsf{PRP}_o(k,\cdot)$ is a permutation.

Second, for any polynomially-bounded $o=o(\lambda)$ and any efficient adversary $\mathcal{A}$, $\mathcal{A}$ cannot distinguish $\textsf{PRP}_o(k,\cdot)$ for a random $k\in_R G_\lambda$ from a uniformly random permutation $P\in_R \mathcal{P}_{G_o\rightarrow G_o}$. We consider the SPRP variant where $\mathcal{A}$ has access to both $P$ and $P^{-1}$. That is, there exists a negligible $negl(\cdot)$ such that
\begin{footnotesize}
\begin{align*}
\left| \underset{k\in_R G_\lambda}{Pr}\left[ \mathcal{A}^{\textsf{PRP}_o(k,\cdot),\textsf{PRP}_o^{-1}(k,\cdot)}(\lambda)=1 \right] - Pr \left[ \mathcal{A}^{P(\cdot),P^{-1}(\cdot)}(\lambda)=1\right] \right| < negl(\lambda).
\end{align*}
\end{footnotesize}
\end{defn}

We remark that \textsf{(C,Q)} works as before, but now the adversary also gets the same type of query access to the inverse permutations. We also introduced the notation $\textsf{D}$, where $\textsf{D} = \textsf{LargeD}$ means that we only require the security to hold for groups $G$ whose sizes are also lower-bounded by a polynomial and $\textsf{D} = \textsf{AnyD}$ means that security must hold for arbitrarily sized groups (small as well as large.) We often suppress the $o$ in $G_o$, and write $G$, if $o$ is already understood from the context.

Finally, we define the \textit{Group-based Function to Permutation Converters (GFPCs)}.

\begin{defn}
Let $\textsf{Q}\in \lbrace \textsf{CQuery, QQuery}\rbrace$, $\textsf{D}\in \lbrace \textsf{LargeD, AnyD}\rbrace$. Fix a sequence of functions $\mathcal{Q}$ in $o,\lambda$. A $(\textsf{Q,D},\mathcal{Q})$-\textbf{GFPC} is a family of pairs of efficient classical oracle algorithms $\textsf{F}_o,\textsf{R}_o$, where:
\begin{itemize}
\item $\textsf{F}_o,\textsf{R}_o$ take as input $\lambda$ and $x\in G_o$, and output $y\in G_o$,
\item There exist polynomials $m(o,\lambda),n(o,\lambda)$ such that $\textsf{F}_o,\textsf{R}_o$ make queries to a function $O: G_m \rightarrow G_n$,
\item $\textsf{F}_o,\textsf{R}_o$ are efficient, i.e. make at most $poly(o,\lambda)$ queries to $O$,
\item $\textsf{F}_o,\textsf{R}_o$ are inverses: $\textsf{R}_o^O(\lambda, \textsf{F}_o^O(\lambda,x))=x$ for all $x\in G_o$ and all oracles $O$,
\item $\textsf{F}_o,\textsf{R}_o$ are indistinguishable from a random permutation and its inverse, given query access.
\end{itemize}
\end{defn}

The last property is saying that if $o(\lambda)$ is any polynomially-bounded function (with the added lower bound if $\mathsf{D}=\mathsf{LargeD}$) and $q(\lambda)=q(o(\lambda),\lambda)$ is any function in $\mathcal{Q}$, then for any adversary $\mathcal{A}$ making at most $q=q(\lambda)$ queries (either classical or quantum, depending on $\mathsf{Q}$), there exists a negligible function $negl(\lambda)$ such that
\begin{align*}
\left| Pr \left[ \mathcal{A}^{F_o^{O}(\lambda,\cdot),R_o^{O}(\lambda,\cdot)}(\lambda)=1\right] - Pr \left[ \mathcal{A}^{P,P^{-1}}(\lambda)=1 \right] \right| < negl(\lambda).
\end{align*}

We remark here that we are unsure whether a GFPC lies in the construction of pseudorandom permutations on groups as we are simply generalizing Zhandry's results rather than his investigation. Our group-based constructions (EM, Feistel) used PRFs from $G$ to $G$ rather than from one group to another, which might of course be a possibility in other constructions, so we have allowed for this in the generalization of the FPCs.

\subsubsection{Main Results}
Recall that what we wish to show is, if quantum one-way functions exist, then so do quantum pseudorandom permutations. For this, we rely heavily on the following main lemma.

\begin{lem}\label{Lem50}
Let $\mathcal{Q}$ be the set of all polynomials. If (\textsf{C,Q})-PRFs exist and (\textsf{C,Q},$\mathcal{Q}$)-GFPCs exist, then (\textsf{C,Q,D})-PRPs exist.
\end{lem}

\begin{proof}
In order to prove the lemma, it suffices to show it in the case of the triple $(\textsf{QComp,QQuery, AnyD})$ as the other cases are analogous. We prove it by a hybrid argument.

\begin{itemize}
\item[\textbf{Hybrid 0}:] Give the adversary the permutations $P(x) = PRP_o(k,x) = F_o^{PRF_{m,n}(k,\cdot)}(\lambda,x)$ and $P^{-1}(x) = PRP_o^{-1}(k,x) = R_o^{PRF_{m,n}(k,\cdot)}(\lambda,x)$, where $k\in_R G_\lambda$. This is the base case where the GFPCs employ the PRF.
\item[\textbf{Hybrid 1}:] Give the adversary the permutations $P(x) = F_o^{O(\cdot)}(\lambda, x)$ and $P^{-1}(x) = R_o^{O(\cdot)}(\lambda, x)$, for a uniformly random $O\in_R \mathfrak{F}_{G_m\rightarrow G_n}$. This is our intermediate case, where the PRF is exchanged with a random function.
\item[\textbf{Hybrid 2}:] Give the adversary the permutations $P(x) = R(x)$ and $P^{-1}(x) = R^{-1}(x)$, where $R,R^{-1}$ are uniformly random permutations on $G_o$.
\end{itemize}

For the case between \textbf{Hybrid 0} and \textbf{Hybrid 1}, we make the following reduction. Assume there exists an adversary $\mathcal{A}$ which distinguishes between \textbf{Hybrid 0} and \textbf{Hybrid 1} with non-negligible probability. Construct the PRF distinguisher $\mathcal{A}^{f}(\lambda) = \mathcal{A}^{F_o^{f(\cdot)}(\lambda,\cdot),R_o^{f(\cdot)}(\lambda,\cdot)}(\lambda)$, where $f:G_m \rightarrow G_n$. If $f(\cdot) = PRF_{m,n}(k,\cdot)$, then $\mathcal{A}$'s view is identical to \textbf{Hybrid 0}. If $f(\cdot) = O(\cdot)$, then $\mathcal{A}$'s view is identical to \textbf{Hybrid 1}. For any quantum query $\mathcal{A}$ asks its oracle, $\mathcal{A}$ may reply using only a polynomial amount of queries to its own oracle, by the third property of GFPCs. The advantage of $\mathcal{A}$ distinguishing between the PRF and a random function is equal to $\mathcal{A}$'s distinguishing possibility i.e. non-negligible, a contradiction. We therefore get that for some negligible function $negl'(\cdot)$,
\begin{footnotesize}
\begin{align*}
&\left| \underset{k\in_R G_{\lambda}}{Pr}\left[ \mathcal{A}^{F_o^{PRF_{m,n}(k,\cdot)}(\lambda,\cdot),R_o^{PRF_{m,n}(k,\cdot)}(\lambda,\cdot)}(\lambda) = 1 \right] - Pr \left[ \mathcal{A}^{F_o^{O(\cdot)},R_o^{O(\cdot)}}(\lambda) = 1 \right] \right| < negl'(\lambda).
\end{align*}
\end{footnotesize}
For the case between \textbf{Hybrid 1} and \textbf{Hybrid 2}, we need only look to the last property of GFPCs, which gives us that there exists a negligible function $negl''(\cdot)$ such that
\begin{align*}
\left| Pr \left[ \mathcal{A}^{F_o^{O(\cdot)},R_o^{O(\cdot)}}(\lambda) = 1 \right] - Pr\left[ \mathcal{A}^{R(\cdot),R^{-1}(\cdot)}(\lambda) = 1 \right] \right| < negl''(\lambda).
\end{align*}
Combining these results, we get that
\begin{align*}
\left| \underset{k\in_R G_{\lambda}}{Pr}\left[ \mathcal{A}^{F_o^{PRF_{m,n}(k,\cdot)}(\lambda,\cdot),R_o^{PRF_{m,n}(k,\cdot)}(\lambda,\cdot)}(\lambda) = 1 \right] - Pr\left[ \mathcal{A}^{R(\cdot),R^{-1}(\cdot)}(\lambda) = 1 \right] \right| \\ \leq negl'(\lambda) + negl''(\lambda) = negl(\lambda),
\end{align*}
by a simple application of the triangle inequality.
\end{proof}

By this lemma, we need only show that the respective GFPCs exist (as Theorem~\ref{OWimplyPRF} gives us the existence of the PRFs.) By \citep{PatelRamzanSundaram}, we have a classical query GFPC, which gives us the following.

\begin{lem}
Let $\mathcal{Q}$ be the set of polynomially-bounded functions, then we have that $(\mathsf{CQuery, LargeD}, \mathcal{Q})$-GFPCs exist.
\end{lem}

Using this lemma and Theorem~\ref{OWimplyPRF} implies the following corollary.
\begin{cor}\label{CorToImprove}
For $\mathsf{C}\in \lbrace \mathsf{CComp, QComp} \rbrace$, if $\mathsf{C}$-one-way functions exist, then $(\mathsf{C, CQuery, LargeD})$-PRPs exist.
\end{cor}

Here Zhandry remarks that we could possibly use Feistel networks to get the end result, however, he does not know how to do so. We have therefore tried to do so, but failed in trying to even get a reduction to the Hidden-Shift Problem, though there were many promising attempts. We therefore also encourage any reader to attempt such a proof.

Consider that if the amount of queries was $q=|G_o|$, then any adversary (possibly inefficient) may query the entire domain and image of the GFPC, regardless of query type. In this way, the distinction between classical and quantum GFPCs disappears. As GFPCs were defined for arbitrary adversaries, this poses no problems even if adversaries making all those queries are inefficient. Let us therefore introduce the following definition.

\begin{defn}
The family of pairs $(\mathsf{F}_o,\mathsf{R}_o)$ is a \textbf{$\mathsf{D}$-Full Domain GFPC} if it is a $(\mathsf{CQuery, D},\mathcal{Q})$-GFPC where $\mathcal{Q}$ contains a function $q(o,\lambda)$ such that $q(o,\lambda)\geq |G_o|$.
\end{defn}

If an adversary has access to the entire output table of $\mathsf{F}_o$ and $\mathsf{R}_o$, then the type of queries it makes is irrelevant, hence, we get the following lemma.

\begin{lem}\label{Lem54}
Let $\mathcal{Q}$ be any class of functions. Then a $\mathsf{D}$-Full Domain GFPC is also a $(\mathsf{QQuery, D},\mathcal{Q})$-GFPC.
\end{lem}

The following lemma is given by Zhandry under the explanation that several articles\footnote{For a full list refer to \citep{Zhandry}.}, centered around card shuffles, give the existence of Full-Domain FPCs. Especially three articles build upon each other: \citep{HMR, RY13, MR14}. In \citep{HMR}, Hoang, Morris, and Rogaway introduce a \textit{Swap-or-Not Shuffle} on bit strings which they show generalizes to arbitrary domains and works for $\mathsf{LargeD}$. In \citep{MR14}, Morris and Rogaway consider the \textit{Mix-and-Cut Shuffle} of \citep{RY13} and refine it into the \textit{Sometimes-Recurse Shuffle}, which they show, when instantiated with the Swap-or-Not Shuffle as an "inner shuffle", gives $\mathsf{AnyD}$ security on arbitrary domains. Using this work, we get the following lemma as long as we may order and encode a group efficiently, i.e. make a bijection from an ordered set of group elements of $G$ to the set of positive integers of size $|G|$, in polynomial-time.

\begin{lem}\label{AnyDexists}
$\mathsf{AnyD}$-Full Domain GFPCs exist.
\end{lem}

This will suffice for our generalization of Zhandry's result, however, we attempt to generalize the shuffles to work on arbitrary groups, without any encoding, in the next section. This is because we feel that the encoding sacrifices the structure of the group, which we believe may be interesting to preserve. Regardless, what these articles intuitively do is make a good PRP and then make it even better, using Lemma~\ref{AnyDexists}. This leads to an improvement of Corollary~\ref{CorToImprove}.

\begin{cor}
For $\mathsf{C}\in \lbrace \mathsf{CComp,QComp}\rbrace$, if $\mathsf{C}$-one-way functions exist, then $(\mathsf{C,CQuery,AnyD})$-PRPs exist.
\end{cor}

Combining Lemma~\ref{OWimplyPRF}, Lemma~\ref{Lem50}, Lemma~\ref{AnyDexists}, and Lemma~\ref{Lem54}, we get Zhandry's main result, but now for arbitrary groups.

\begin{thm}
The existence of $\mathsf{QComp}$-one-way functions implies the existence of $(\mathsf{QComp, QQuery, AnyD})$-PRPs.
\end{thm}

This result shows that even the concept of an FPC is generalizable to arbitrary groups and reenforces the idea that strong results on bit strings could hold on arbitrary groups. Based on this result we might already consider that groups could be of use in quantum cryptography, especially given that the results of Kuwakado and Morii \citep{ BrokenFeistel, BrokenEM} seem to only apply to bit strings and not necessarily to arbitrary groups (see \citep{Gorjan}.)

\newpage
\section{Generalizing Card Shuffles}\label{ShuffleSection}
In the previous section, we reproduced Zhandry's result over arbitrary groups but assumed the existence of an efficient ordering of group elements and encoding to the integers in order to do so. In this section, we consider whether the card shuffles of \citep{HMR} and \citep{MR14} may be generalized to arbitrary groups, so that we may discard the assumption. The first shuffle construction, \textit{Swap-or-Not}, leads to the creation of a new shuffle we choose to call \textit{Scoot-or-Not}, due to the fact that it no longer swaps elements, but rather "scoots"\footnote{As formulated by Gorjan Alagic} the elements. Despite spending time and effort, the second shuffle construction, \textit{Sometimes-Recurse}, does not quite generalize, for reasons we briefly discuss.

\subsection{Swap-or-Not and Scoot-or-Not Shuffles}
The Swap-or-Not (SN) shuffle  is a new way to turn a PRF into a PRP. It can do so in a bit string format or even an arbitrary finite domain size, as long as the domain is ordered. This is useful in order to encipher information that is not necessarily in a $2^n$ sized domain, e.g. credit-card numbers which are usually in a $10^n$ sized domain. Black and Rogaway \citep{BlackRogaway} also considered this problem but only gave enciphering schemes for "small" and "large" domains. The beauty of the SN shuffle is that it works over any arbitrarily sized domain and even allows an adversary to query almost the entire domain!

Let us start by briefly explaining the terminology used in this section, in terms of an encryption scheme. Consider a \textit{deck} $\mathcal{X}$ of $N\in \mathbb{N}$ cards, which is our domain. For example, $\mathcal{X}$ could be $\lbrace 0,1 \rbrace^n$ for $N=2^n$, $\mathcal{X}$ could be $\left[ N \right]$, i.e. the group of integers with addition modulo $N$, or $\mathcal{X}$ could be some other group $G$. A \textit{shuffle} is then a permutation on the domain, or rather, the block-cipher used to permute the enciphering elements. We may even \textit{cut} the deck, as one would in a game of poker, which would have the analogue of taking subsets, or perhaps subgroups. The randomness used in the shuffling of the deck can be considered the random key.

In this section, we use the concept of an \textit{oblivious} shuffle, originally suggested by Moni Naor. In such a shuffle, we need only consider the trace of a single card, rather than trace all, or a lot, of the other cards, in order to prove security.

The SN shuffle itself works in a series of "rounds" where elements in a domain are swapped, or not, and the output of the shuffle is then the permuted domain. In Figure~\ref{SNRound} and Figure~\ref{abelianSN}, we give the definition of a single round (in the bit string case) and the generalized construction in a finite abelian group, respectively.

\begin{figure}[h!]
\centering
\begin{pseudocode}[framebox]{$SN\left[1,\left( \mathbb{Z}/2\mathbb{Z}\right)^n,\oplus\right]$}{X}
K \in_R \lbrace 0,1 \rbrace^n \\
\FOR \text{the pair of positions } \lbrace X, K \oplus X \rbrace \DO \BEGIN
b\in_R \lbrace 0,1 \rbrace \\
\IF b=1 \THEN \text{swap the cards at positions } X \text{ and } K \oplus X \END
\end{pseudocode}
\caption{A Single Round Swap-or-Not Shuffle}
\label{SNRound}
\end{figure}

In Figure~\ref{SNRound}, we use $\mathcal{X}=\left( \mathbb{Z}/2\mathbb{Z}\right)^n$, where each card is $X\in \lbrace 0,1 \rbrace^n$, i.e. the bit strings. At first a random key $K\in_R \lbrace 0,1 \rbrace^n$ is selected for the round. Each $X$ is paired up with $K\oplus X$, which is a unique pairing by the abelian property of $\left( \mathbb{Z}/2\mathbb{Z}\right)^n$ under the XOR operation, $\oplus$. A random bit is then used to decide whether the pair is swapped or not. The shuffle can be repeated for more rounds where each round uses an independently chosen random key and independently chosen random bit $b$. The random bits may be instantiated with random round functions $F_i:\lbrace 0,1 \rbrace^n \rightarrow \lbrace 0,1 \rbrace$, for each round $i$. The shuffle is repeated for all $X\in \mathcal{X}$ using the same $K$ and $b$. We now give the $r$-round SN shuffle over an abelian group.

\begin{figure}[h!]
\centering
\begin{pseudocode}[framebox]{$SN\left[r,\left[N\right],+\right]$}{X}
\FOR i:= 1 \TO r \DO 
\BEGIN X' := K_i - X, \hat{X}:= \max(X,X') \\ 
\IF F_i(\hat{X})=1 
\THEN \text{ redefine } X := X' 
\END \\
\RETURN X
\end{pseudocode}
\caption{abelian Group Swap-or-Not Shuffle for $r$ rounds}
\label{abelianSN}
\end{figure}

In Figure~\ref{abelianSN} we consider the abelian group $\left[ N \right] = \mathbb{Z}/N\mathbb{Z}$ with addition modulo $N$ as its operation. The keys $K_i\in \left[ N \right]$ are chosen independently at random, but only differ from round to round, they are not rechosen for each element $X\in \left[ N \right]$. Likewise, the $F_i:\left[ N \right] \rightarrow \lbrace 0,1 \rbrace$ are independently chosen round functions, that vary only per round, not per element $X\in \left[ N \right]$. When repeated on each element $X\in \left[ N \right]$, this shuffle permutes the elements of the deck and has an inverse, which is gotten by reversing the rounds, i.e. going from round $r$ to round $1$. Using $SN\left[r,\left[N\right],+\right]$, \citep{HMR} show that in the security setting, an adversary may distinguish between being given the shuffle and a random permutation with advantage at most
\begin{align*}
\frac{8N^{3/2}}{r+4}\left( \frac{q+N}{2N}\right)^{r/4+1},
\end{align*}
when limited to at most $q$ total queries. For good security, the above bound allows only up to $q=(1-\varepsilon)N$ adversarial queries, where $\varepsilon > 0$ is some fixed value, which is why complete information theoretic security escapes us. Hoang, Morris, and Rogaway also consider "tweaking" their shuffle, which we have not considered generalizing.

The SN shuffle itself has two main properties: the same algorithm applies the permutation and the inverse permutation, by simply reversing the rounds, and the elements are uniquely swapped in every round.

Let us try to make a direct generalization of the SN shuffle to an arbitrary group $G$. We first notice that the $\max$ function needs an ordering on the elements, which is not a problem on the integers but might be for other groups (consider very large groups or ones with complex structures.) As the authors remark as well, the choice of the $\max$ function is not entirely important and could have been exchanged with some other function able to pick a canonical element in the forward and inverse permutation algorithms. Let us now consider the generalized pairing function, $k\cdot x^{-1}$. In $G$, we cannot be certain that elements are swapped. Consider an element $x\in G$ and a key $k\in_R G$. Let us say that the algorithm pairs $x$ with $k \cdot x^{-1}$. At the same time it also pairs $k\cdot x^{-1}$ with $k\cdot (k \cdot x^{-1})^{-1} = k \cdot x \cdot k^{-1}$, which, unless $G$ is abelian, will not equal $x$ in most cases. This also means that we cannot use the same pairing function for the forward permutation and the inverse permutation. As we are not able to directly generalize the shuffle, we use these considerations and construct the following \textit{Scoot-or-Not Shuffle (SC)}, see Figure~\ref{ScootShuffle}.

\begin{center}
\begin{figure}[h]
\fbox{
\begin{minipage}{0.46\textwidth}
\centering
\begin{pseudocode}{$SC\left[r,G,\cdot \right]$}{x}
\FOR i:= 1 \TO r \DO 
\BEGIN x' := k_i \cdot x, \hat{x}:= x'\cdot x^{-1} \\ 
\IF f_i(\hat{x})=1 
\THEN \text{ redefine } x := x' 
\END \\
\RETURN x
\end{pseudocode}
\end{minipage}}
\fbox{
\begin{minipage}{0.46\textwidth}
\centering
\begin{pseudocode}{$SC^{-1}\left[r,G,\cdot \right]$}{x}
\FOR i:= 1 \TO r \DO 
\BEGIN x' := k_i^{-1} \cdot x, \hat{x}:= x'\cdot x^{-1} \\ 
\IF f_i(\hat{x}^{-1})=1 
\THEN \text{ redefine } x := x' 
\END \\
\RETURN x
\end{pseudocode}
\end{minipage}}
\caption{Scoot-or-Not Shuffle and its Inverse Shuffle}
\label{ScootShuffle}
\end{figure}
\end{center}

We remark that the $k_i\in G$ and round functions $f_i:G\rightarrow \lbrace 0,1 \rbrace$ for each $i$ are uniformly chosen at random. In our definition of the SC shuffle, we note that our use of $\hat{x}$ is superfluous as $\hat{x}= x'\cdot x^{-1} = k_i$ and likewise $k_i^{-1}$ in the inverse shuffle. This means that regardless of the value of $x$, whether we "scoot" in a round depends solely on the key, and the shuffle does so for every element when in that round. We note that the SC shuffle can be thought of as a "selective" Group Even-Mansour cipher with up to $r$ keys. We mean this in the sense that, if, for example, only two round functions, $f_i$ and $f_j$, choose to scoot $x$, then $x$ gets sent to $k_j\cdot (k_i\cdot x)$ in the shuffle.

The two shuffles in Figure~\ref{ScootShuffle} are indeed inverses. Denote $E_{KF}(x) := SC\left[r,G,\cdot \right](x)$ and $D_{KF}(x) := SC^{-1}\left[r,G,\cdot \right](x)$, where $KF$ is the set of keys and round functions used in the shuffles. Let us show that $D_{KF}(E_{KF}(x))=x$ in the one round case.

Consider $x\in G$, $k\in_R G$ and a random round function $f:G \rightarrow \lbrace 0,1 \rbrace$. From $E_{KF}(x)$, we get that $x'=k\cdot x$ and $\hat{x} = x'\cdot x^{-1} = k$. The round function $f$ now gives us two possibilities.
\begin{align*}
E_{KF}(x) = \begin{cases} x & \text{ if } f(\hat{x})=0, \\
				  k\cdot x & \text{ if } f(\hat{x})=1.	 \end{cases}
\end{align*}
If $E_{KF}(x)=x$, i.e. $ f(k)=0$, then from $D_{KF}(E_{KF}(x))$, we get $x'= k^{-1}\cdot x$ and $\hat{x}=k^{-1}\cdot x \cdot x^{-1} = k^{-1}$. As $f(\hat{x}^{-1}) = f((k^{-1})^{-1}) = f(k) = 0$, we get that $D_{KF}(E_{KF}(x))=x$. If instead $E_{KF}(x)=k\cdot x$, i.e. $ f(k)=1$, then from $D_{KF}(E_{KF}(x))$, we get $x'= k^{-1}\cdot (k\cdot x) = x$ and
\begin{align*}
\hat{x}=x \cdot (k\cdot x)^{-1} = x\cdot x^{-1} \cdot k^{-1} = k^{-1}.
\end{align*}
As $f(\hat{x}^{-1}) = f((k^{-1})^{-1}) = f(k) = 1$, we get that $D_{KF}(E_{KF}(x))=x$. This also holds for $r$ rounds and the reversed case is analogous.

\subsubsection*{Preliminaries}
We wish to prove some form of security for our shuffle, let us therefore consider chosen-ciphertext attack (CCA) security. Let us again consider the positive integer $\lambda$ to be the security parameter, specified in unary per convention. We assume that for each $\lambda$ there exists a uniquely specified group $G(\lambda) = G_\lambda \in \mathcal{G}$ with size $|G_\lambda| \geq 2^\lambda$.

\begin{defn}
Let $P:G_\lambda \times G \rightarrow G$ be a blockcipher on a group $G\in \mathcal{G}$, such that $P_k(\cdot)=P(k,\cdot)$ is a permutation on $G$, with inverse blockcipher $P^{-1}$. For any adversary $\mathcal{A}$, we define the \textbf{CCA advantage} of $\mathcal{A}$ in carrying out an (adaptive) chosen-ciphertext attack on $P$ as
\begin{align*}
Adv_P^{cca}(\mathcal{A}) = \left| \underset{k \in_R G_\lambda}{Pr}\left[ \mathcal{A}^{P_k(\cdot),P_k^{-1}(\cdot)}(\lambda)=1 \right] - \underset{\pi \in_R \mathfrak{P}_{G\rightarrow G}}{Pr}\left[ \mathcal{A}^{\pi(\cdot),\pi^{-1}(\cdot)}(\lambda)=1 \right] \right|,
\end{align*}
where $\mathfrak{P}_{G\rightarrow G}$ is the set of permutations on $G$.
\end{defn}

If the adversary asks no queries to the inverse oracle, then it is carrying out an (adaptive) chosen-plaintext attack (CPA). If the adversary asks the same queries on every run, we call the adversary \textit{non-adaptive}. Define $Adv_P^{cca}(\mathcal{A},q)$ to be the maximum advantage of an adversary $\mathcal{A}$, asking at most $q$ total queries. Likewise, we define $Adv_P^{ncpa}(\mathcal{A},q)$ to be the maximum advantage of an adversary for nonadaptive CPA attacks (NCPA).

If we assume $L,M:G_\lambda \times G \rightarrow G$ to be blockciphers, then we may let $L\circ M$ denote the composed blockcipher $(L\circ M)_{(K,K')}(X)=L_K(M_{K'}(X))$. We refer to Maurer, Pietrzak, and Renner \citep{Maurer} for the following result.

\begin{lem}[Maurer-Pietrzak-Renner]\label{MaurerPietrzakRenner}
If $L$ and $M$ are blockciphers on the same message space then, for any adversary $\mathcal{A}$ and any $q$,
\begin{align*}
Adv_{M^{-1}\circ L}^{cca}(\mathcal{A},q) \leq Adv_L^{ncpa}(\mathcal{A},q) + Adv_M^{ncpa}(\mathcal{A},q).
\end{align*}
\end{lem}

We will also need the definition of \textit{total variation distance}.
\begin{defn}
Let $\mu$ and $\nu$ be probability distributions on $\Omega$. The \textbf{total variation distance} between distribution $\mu$ and $\nu$ is defined as
\begin{align*}
\norm{\mu - \nu} = \frac{1}{2}\sum_{x\in \Omega} \left|\mu(x) - \nu(x)\right| = \underset{S\subset \Omega}{\max} \lbrace \mu(S) - \nu(S)\rbrace .
\end{align*}
\end{defn}

\subsubsection*{Security of Scoot-or-Not}
Fix a group $G\in \mathcal{G}$ with operation $\cdot$, where $|G| = N \in \mathbb{N}$. We will now consider the $r$-round Scoot-or-Not shuffle $SC\left[ r, G, \cdot \right]$ where we instantiate round $i$ of the shuffle using uniformly random key $k_i\in_R G$ and uniformly random bit $b_i\in_R \lbrace 0,1 \rbrace$ in place of the round function, i.e. we scoot if $b_i=1$.

We may consider this shuffle as a Markov chain $\lbrace W_i | i\geq 0\rbrace$ as each round $i$ is independent of the preceding and proceeding rounds. In this way, we consider the state space of $\lbrace W_i\rbrace$ to be the set of elements, here called cards, of $G$ and interpret $W_i(x)$ to be the position of the card $x$ at time $i$.

In the following, we assume that our adversary $\mathcal{A}$ is a non-adaptive adversary, i.e. deterministic, making exactly $q\leq |G|$ queries. We then analyse the mixing rate of the shuffle. Since $\mathcal{A}$ only makes $q$ queries, we need only bound the mixing rate for some $q$ sized subset of the deck. Let $x_1,\cdots,x_q$ be the set of distinct cards and let $Z_i(x_j)$ be the vector position of the card $x_j$ at time $i$, defining $Z_i(x_1,\cdots,x_j)=(Z_i(x_1),\cdots , Z_i(x_j))$. We call the set of $Z_i$, for $1\leq i \leq r$, the \textit{projected Scoot-or-Not shuffle}. Let $\pi$ denote the stationary distribution of $Z_i$, which is uniform over the set of distinct $q$-tuples of cards, i.e. the uniform distribution of $q$ samples without replacement. We denote $\tau_i$ the distribution of $Z_i$ and show the following.

\begin{thm}[Rapid Mixing]
Consider the Scoot-or-Not shuffle $SC\left[ r, G, \cdot \right]$ for $r,|G|\geq 1$, and let $1\leq q \leq |G|$. Fix distinct $x_1,\ldots, x_q$ and let $\lbrace Z_i | i \geq 0\rbrace$ be the corresponding projected Scoot-or-Not shuffle, let $\pi$ be its stationary distribution, and let $\tau_i$ be the distribution of $Z_i$. Then
\begin{align*}
\norm{\tau_r - \pi} \leq \frac{2|G|^{3/2}}{r+2}\left( \frac{q+|G|}{2|G|}\right)^{r/2+1}.
\end{align*}
\end{thm}

\begin{proof}
We start by doing some probability theoretic considerations. We let $\tau_i^k$ be the conditional distribution of $Z_i$ given the random keys $k_1,\cdots,k_r$. Since the $k_i$ are random variables, the same will hold for $\tau_i^k$, such that also $\norm{\tau_i^k - \pi}$ is. As $\tau_r = E(\tau_r^k)$, where $E$ signifies the probabilistic expectation, we get
\begin{align*}
\norm{\tau_r - \pi} = \norm{E(\tau_r^k - \pi)} \leq E(\norm{\tau_r^k - \pi}),
\end{align*}
where the inequality stems from Jensen's inequality as the total variation distance is half of the $\mathcal{L}^1$-norm, which is convex. 

Let $\nu$ be a distribution on $q$-tuples of a set $\Omega$, then we may define
\begin{align*}
\nu(u_1,\cdots,u_j) &= Pr\left[ Y_1=u_1, \cdots , Y_j = u_j \right], \\
\nu(u_j | u_1,\cdots, u_{j-1}) &= Pr\left[ Y_j = u_j | Y_1 = u_1, \cdots, Y_{j-1}=u_{j-1}\right],
\end{align*}
where $(Y_1,\cdots,Y_q)\sim \nu$. In the SC shuffle, $\tau_i(u_1,\cdots,u_j)$ is the probability of the $j$ first cards $x_1,\cdots,x_j$ end up in the positions $u_1,\cdots,u_j$ at time $i$ and $\tau_i(u_j|u_1,\cdots,u_{j-1})$ is the probability of $x_j$ being in spot $u_j$ at time $i$ when the other cards are in given positions. An analogous definition holds for the uniformly random $\pi$. Using this notation, we present the following lemma, whose proof we omit, but a slightly more general proof may be found in \citep{MR14}.

\begin{lem}\label{NormEstimate}
Fix a finite nonempty set $\Omega$ and let $\mu$ and $\nu$ be probability distributions supported on $q$-tuples of elements of $\Omega$, and suppose that $(Y_1,\cdots,Y_q)\sim \mu$, then
\begin{align}\label{ProbabilityEstimate}
\norm{\mu-\nu} \leq \sum_{l=0}^{q-1}E\left( \norm{\mu(\cdot | Y_1,\cdots,Y_l)-\nu(\cdot | Y_1,\cdots, Y_l)}\right).
\end{align}
\end{lem}

Since the $Y_i$ are random variables, so is $\norm{\mu(\cdot | Y_1,\cdots,Y_l)-\nu(\cdot | Y_1,\cdots, Y_l)}$ for each $l< q$.

Fix an integer $l\in \lbrace 0,\cdots, q-1\rbrace$ such that we consider the set of cards $\lbrace x_1,\cdots, x_l\rbrace$ and aim to bound the expected distance between the conditional distributions for $\tau_i^k$ and $\pi$. For $i\geq 0$, we therefore consider the set $S_i = G \setminus \lbrace Z_i(x_1), \cdots Z_i(x_l)\rbrace$, i.e. the set of positions that card $x_{l+1}$ can be in at time $i$ given the positions of cards $x_1,\cdots, x_l$ at time $i$. For $a\in S_i$ we introduce the notation $p_i(a)=\tau_i^k(a | Z_i(x_1), \cdots Z_i(x_l))$. For $m = |S_i| = |G|-l$, we have
\begin{align*}
\norm{\tau_i^k( \cdot | Z_i(x_1,\cdots,x_l)) - \pi( \cdot | Z_i(x_1,\cdots,x_l))} = \frac{1}{2}\sum_{a\in S_i} |p_i(a)-1/m |,
\end{align*}
for which we get the following estimate,
\begin{align}
\left( E\left( \sum_{a\in S_i} |p_i(a)-1/m |\right) \right)^2 &\leq E\left( \left( \sum_{a\in S_i} |p_i(a)-1/m |\right)^2 \right) \nonumber \\
&\leq m \cdot E\left( \sum_{a\in S_i} (p_i(a)-1/m )^2 \right) \nonumber \\
&\leq |G| \cdot E\left( \sum_{a\in S_i} (p_i(a)-1/m )^2 \right), \label{NEstimate}
\end{align}
by two applications of Cauchy-Schwarz and using that $m\leq |G|$. We now wish to show by induction that
\begin{align}\label{HMREstimate}
E\left( \sum_{a\in S_i} (p_i(a)-1/m )^2 \right) \leq \left( \frac{l+|G|}{2|G|}\right)^i,
\end{align}
for every $i\leq r$. For the induction start, $i=0$, consider that the initial card positions are deterministic, such that
\begin{align*}
E\left( \sum_{a\in S_0} (p_i(a)-1/m )^2 \right) = (1-1/m)^2 + (m-1)\cdot (0-1/m)^2 = 1 - 1/m < 1.
\end{align*}
This is because $S_0$ is the set of positions that card $x_{l+1}$ could be located in at time $0$, given the positions of cards $x_1,\cdots , x_l$. Since $p_0(a)$ is the probability that $x_{l+1}$ is in position $a$ at time $0$ and because the positions of the cards at time $0$ are deterministic, we know exactly which $a$ the card $x_{l+1}$ is sent to.
Let us now consider the induction step and assume that Equation (\ref{HMREstimate}) holds for $i$. Define $s_i = \sum_{a\in S_i} (p_i(a)-1/m )^2$. As $E(E(X|Y))=E(X)$, for independent $X$ and $Y$, and $E(cX)=cE(X)$, for a constant $c$, we need only show that
\begin{align*}
E(s_{i+1}|s_i) = \left( \frac{l+|G|}{2|G|}\right)s_i.
\end{align*}
Define the auxilliary function $h:S_i \rightarrow S_{i+1}$,
\begin{align*}
h(a) = \begin{cases} a & \hspace*{10pt} \text{ if } a\in S_{i+1}, \\ k_{i+1}\cdot a & \hspace*{10pt} \text{ else. } \end{cases}
\end{align*}
This function is well-defined: per definition of the shuffle, if $a\in S_i$ then either $a\in S_{i+1}$ or $k_{i+1}\cdot a \in S_{i+1}$. It is also bijective\footnote{For bijectivity argument, see Appendix~\ref{Bijectivityofh}.}, having the inverse,
\begin{align*}
h^{-1}(b) = \begin{cases} b & \hspace*{10pt} \text{ if } b\in S_{i}, \\ k_{i+1}^{-1}\cdot b & \hspace*{10pt} \text{ else. } \end{cases}
\end{align*}

Notice also that because of the uniformity of the round function (here instantiated by a random bit,) we get
\begin{align*}
p_{i+1}(h(a)) = \begin{cases} p_i(a) & \hspace*{10pt} \text{ if } k_{i+1}\cdot a\not\in S_{i}, \\ \tfrac{1}{2}p_i(a) + \tfrac{1}{2}p_i(k_{i+1}\cdot a) & \hspace*{10pt} \text{ else. } \end{cases}
\end{align*}

We note that
\begin{align*}
E(s_{i+1}|s_i) = \sum_{b\in S_{i+1}} E\left( \left(p_{i+1}(b)-1/m\right)^2 | s_i\right) = \sum_{a\in S_{i}} E\left( \left(p_{i+1}(h(a))-1/m\right)^2 | s_i\right),
\end{align*}
so we aim to bound the summands of the latter sum.

As $k_{i+1}$ is independent of the other rounds, we have that
\begin{align*}
Pr\left[ k_{i+1}\cdot a = y | y\in G \right] = 1/|G|.
\end{align*}
Recall that $m=|S_i|=|G|-l$. By conditioning on the value of $k_{i+1}\cdot a$ we get that
\begin{footnotesize}
\begin{align*}
E\left( \left( p_{i+1}(h(a))-1/m\right)^2 | s_t \right) = \frac{l}{|G|}(p_i(a)-1/m)^2 + \frac{1}{|G|}\sum_{y\in S_i}\left(\frac{p_i(a)+p_i(y)}{2}- \frac{1}{m}\right)^2,
\end{align*}
\end{footnotesize}
where the sum may be rewritten as
\begin{align*}
&\hspace*{10pt} \sum_{y\in S_i}\tfrac{1}{4}\left( (p_i(a)-1/m)+ (p_i(y)-1/m) \right)^2 \\
&= \tfrac{1}{4} \sum_{y\in S_i} (p_i(a)-1/m)^2 + \tfrac{1}{4} \sum_{y\in S_i} (p_i(y)-1/m)^2 + \tfrac{1}{2}(p_i(a)-1/m)\sum_{y\in S_i}(p_i(y)-1/m) \\
&= \tfrac{m}{4}(p_i(a)-1/m)^2 + \tfrac{1}{4}s_i,
\end{align*}
where we used that $\sum_{y\in S_i} (p_i(y) - 1/m)=0$. Hence we get that
\begin{align*}
E\left( \left( p_{i+1}(h(a))-1/m\right)^2 | s_t \right) = \frac{4l+m}{4|G|}(p_i(a)-1/m)^2 + \frac{s_i}{4|G|},
\end{align*}
which further gives us, by evaluating at each $a\in S_i$, that
\begin{align*}
E(s_{i+1}|s_i) &= \frac{4l+m}{4|G|}\sum_{a\in S_i}(p_i(a)-1/m)^2 + \frac{ms_i}{4|G|} \\
	&= \frac{4l+m}{4|G|}s_i + \frac{ms_i}{4|G|} \\
	&= \frac{l+|G|}{2|G|}s_i,
\end{align*}
such that we have proven our induction step.

Let us now use this to conclude our theorem. Letting $i=r$ in the above, we get that
\begin{align*}
\norm{\tau_r - \pi} &\leq E(\norm{\tau_r^k - \pi})  \\
	&\leq \sum_{l=0}^{q-1} E \left(\norm{\tau_i^k( \cdot | Z_i(x_1,\cdots,x_l)) - \pi( \cdot | Z_i(x_1,\cdots,x_l))}\right) && \text{by Lemma~\ref{NormEstimate}}\\
	&\leq \sum_{l=0}^{q-1} \frac{1}{2} \left( |G| \cdot E \left( \sum_{a\in S_r} (p_r(a)-1/m )^2\right) \right)^{1/2} && \text{by (\ref{NEstimate})} \\
	&\leq \sum_{l=0}^{q-1} \frac{\sqrt{|G|}}{2}\left( \frac{l+|G|}{2|G|}\right)^{r/2} && \text{by (\ref{HMREstimate})} \\
	&\leq \frac{\sqrt{|G|}}{2}\int_0^{q-1} \left( \frac{1}{2} + \frac{l}{2|G|}\right)^{r/2} dl \\
	&= \frac{\sqrt{|G|}}{2}\int_0^{\tfrac{q-1}{2|G|}}\left( \frac{1}{2}+ x \right)^{r/2}\cdot 2|G| dx \\
	&\leq |G|^{3/2} \int_0^{\tfrac{q}{2|G|}} \left( \frac{1}{2}+ x\right)^{r/2} dx \\
	&= |G|^{3/2} \int_{\tfrac{1}{2}}^{\tfrac{q}{2|G|}+\tfrac{1}{2}}  y^{r/2} dy \\
	&= \frac{2|G|^{3/2}}{r+2} \left( \left( \frac{q}{2|G|}+ \frac{1}{2} \right)^{r/2+1}- \left(\frac{1}{2}\right)^{r/2+1}\right) \\
	&\leq \frac{2|G|^{3/2}}{r+2}\left( \frac{q+|G|}{2|G|}\right)^{r/2+1}.
\end{align*}
\end{proof}

We note that the above estimate possibly may be reduced as it uses rather large estimates. However, for our purposes, this estimate is sufficient and also equivalent to the estimate given in \citep{HMR}. By an analogous argument, it is possible to prove that the inverse SC shuffle satisfies the exact same bound. Using Lemma~\ref{MaurerPietrzakRenner}, we may therefore obtain CCA security as well.

\begin{thm}
Let $P= SC\left[ 2r, G, \cdot \right]$, then
\begin{align*}
Adv_P^{cca}(\mathcal{A},q) \leq \frac{4|G|^{3/2}}{r+2}\left( \frac{q+|G|}{2|G|}\right)^{r/2+1}.
\end{align*}
\end{thm}
\begin{proof}
Consider $P=SC\left[ r, G, \cdot \right]$ for the group $G \in \mathcal{G}$, and $P^{-1} = SC^{-1}\left[ r, G, \cdot \right]$, then Lemma~\ref{MaurerPietrzakRenner} gives us that
\begin{align*}
Adv_{(P^{-1})^{-1}\circ P}^{cca}(\mathcal{A},q) &\leq Adv_{P}^{ncpa}(\mathcal{A},q) + Adv_{P^{-1}}^{ncpa}(\mathcal{A},q) \\
	&\leq \frac{4|G|^{3/2}}{r+2}\left( \frac{q+|G|}{2|G|}\right)^{r/2+1}.
\end{align*}
As $(P^{-1})^{-1}$ is $P^{-1}$ reversed, i.e. $P$, the above is equivalent to using $P$ with $2r$ rounds instead of $r$ rounds.
\end{proof}

\subsection{Generalizing the Sometimes-Recurse Shuffle}
The \textit{Sometimes-Recurse} shuffle (SR) is a shuffle which essentially takes a weak PRP and turns it into a stronger PRP. Specifically, the SR shuffle takes an \textit{inner shuffle}, e.g. the SN shuffle, and uses it to shuffle the cards. Then it cuts the deck in some predefined \textit{split}, e.g. in half, and shuffles one of the new smaller decks. It then continues the procedure recursively until there is only one card left to shuffle which per definition is already shuffled. The security proof is based on the security of the inner shuffle which we assume to shuffle some half amount of the cards in the deck (or the recursive deck) to within $\varepsilon >0$ of uniform. In the case of the deck of cards being the set of integers $\lbrace 1,\cdots, N\rbrace$, whenever we make a split of the deck, we may reenumerate the cards to be the sets $\lbrace 1,\cdots, N/2 \rbrace$ and $\lbrace (N+1)/2, \cdots, N\rbrace$, and so forth. In the arbitary group setting, simply splitting into smaller sets give us problems with whether the sets are subgroups. If we consider using a bijection into some smaller group, which is not necessarily a subgroup, then we might as well do this to begin with, but then there is no reason to consider a shuffle on the group itself as we may use the SR and SN shuffle on the initial set of group elements by giving it a bijection to the set of integers of size $|G|$, i.e. an encoding.

We leave as an open question the construction of a shuffle turning our SC shuffle into a PRP which is secure under $q=|G|$ queries by an adversary.

\newpage
\section{Conclusion}\label{ConclusionSection}
We generalized the Even and Mansour scheme as well as the Feistel cipher to work over arbitrary groups and proved that classical results pertain to the group versions. We also generalized the Function to Permutation Converters of Zhandry to get the existence of quantum-secure PRPs from the existence of quantum-secure OWFs. Lastly, we introduced the Scoot-or-Not shuffle, a cipher based on card shuffles.
 
Due to the work by Alagic and Russell in \citep{Gorjan}, we hope that our results open avenues to proving that classical schemes may be made quantum-secure by generalizing them to certain groups. For further work, we suggest generalizing other classical schemes and using the underlying group structures to do Hidden Shift reductions, especially those schemes considered by Kaplan et al in \citep{Kaplan}. Although it has not been included in this thesis, we have constructed a "tweakable" cipher over arbitrary groups which does seem to have a Hidden Shift reduction. Security proofs and further reductions for this cipher would be the basis of further work.

I would like to thank my thesis advisor Gorjan Alagic for the topic, enlightening questions and answers, as well as the encouragements along the way. I would also like to thank my thesis advisor Florian Speelman for the support, fun discussions, and for letting me be his first Master's thesis student.

\pagebreak
\bibliographystyle{alpha}
\bibliography{Articlebib}

\begin{thebibliography}{KLLNP16}

\bibitem[AR17]{Gorjan}
Gorjan Alagic and Alexander Russell.
\newblock Quantum-secure symmetric-key cryptography based on hidden shifts.
\newblock {\em EUROCRYPT 2017}, 2017.

\bibitem[BR02]{BlackRogaway}
John Black and Phillip Rogaway.
\newblock {\em Ciphers with Arbitrary Finite Domains}, pages 114--130.
\newblock Springer Berlin Heidelberg, 2002.

\bibitem[DF99]{DumbFoot}
D.S. Dummit and R.M. Foote.
\newblock {\em Abstract Algebra}.
\newblock John Wiley and Sons, second edition, 1999.

\bibitem[DKS12]{DKS}
Orr Dunkelman, Nathan Keller, and Adi Shamir.
\newblock Minimalism in cryptography: The {Even-Mansour} scheme revisited.
\newblock {\em EUROCRYPT}, 2012.

\bibitem[DR02]{DingDong}
Yan~Zong Ding and Michael~O. Rabin.
\newblock Hyper-encryption and everlasting security.
\newblock In {\em STACS 2002, 19th Annual Symposium on Theoretical Aspects of
  Computer Science, Antibes - Juan les Pins, France, March 14-16, 2002,
  Proceedings}, pages 1--26, 2002.

\bibitem[EM97]{EM}
Shimon Even and Yishay Mansour.
\newblock A construction of a cipher from a single pseudorandom permutation.
\newblock {\em Cryptology}, 1997.

\bibitem[GGM86]{GGM86}
Oded Goldreich, Shafi Goldwasser, and Silvio Micali.
\newblock How to construct random functions.
\newblock {\em Journal of the ACM}, 4(33):792--807, 1986.

\bibitem[GR04]{GentryRamzan}
Craig Gentry and Zulfikar Ramzan.
\newblock Eliminating random permutation oracles in the {Even-Mansour} cipher.
\newblock {\em ASIACRYPT}, 2004.

\bibitem[HILL99]{HILL99}
Johan Håstad, Russell Impagliazzo, Leonid~A. Levin, and Michael Luby.
\newblock A pseudorandom generator from any one-way function.
\newblock {\em SIAM Journal on Computing}, 4(28):1364--1396, 1999.

\bibitem[HMR12]{HMR}
Viet~Tung Hoang, Ben Morris, and Phillip Rogaway.
\newblock An enciphering scheme based on a card shuffle.
\newblock {\em CoRR}, abs/1208.1176, 2012.

\bibitem[Jea16]{TikZhelp}
Jeremy Jean.
\newblock \textup{Ti\textit{k}{Z}} for cryptographers.
\newblock \url{http://www.iacr.org/authors/tikz/}, 2016.

\bibitem[KL15]{KL}
Jonathan Katz and Yehuda Lindell.
\newblock {\em Introduction to Modern Cryptography}.
\newblock CRC Press, 2 edition, 2015.

\bibitem[KLLNP16]{Kaplan}
Marc Kaplan, Gaetan Leurent, Anthony Leverrier, and Maria Naya-Plasencia.
\newblock Breaking symmetric cryptosystems using quantum period finding.
\newblock {\em ArXiv}, 2016.

\bibitem[KM10]{BrokenFeistel}
Hidenori Kuwakado and Masakatu Morii.
\newblock Quantum distinguisher between the 3-round {Feistel} cipher and the
  random permutation.
\newblock In {\em ISIT}, pages 2682--2685. IEEE, 2010.

\bibitem[KM12]{BrokenEM}
Hidenori Kuwakado and Masakatu Morii.
\newblock Security on the quantum-type {Even-Mansour} cipher.
\newblock In {\em ISITA}, pages 312--316. IEEE, 2012.

\bibitem[KR01]{Kilr}
Joe Kilian and Phillip Rogaway.
\newblock How to protect {DES} against exhaustive key search (an analysis of
  {DESX}).
\newblock {\em J. Cryptology}, 14(1):17--35, 2001.

\bibitem[LR88]{LubyR}
Michael Luby and Charles Rackoff.
\newblock How to construct pseudorandom permutations from pseudorandom
  functions.
\newblock {\em SIAM J. Comput.}, 17(2):373--386, 1988.

\bibitem[MPR07]{Maurer}
Ueli Maurer, Krzysztof Pietrzak, and Renato Renner.
\newblock Indistinguishability amplification.
\newblock In {\em Proceedings of the 27th Annual International Cryptology
  Conference on Advances in Cryptology}, CRYPTO'07, pages 130--149, Berlin,
  Heidelberg, 2007. Springer-Verlag.

\bibitem[MR14]{MR14}
Ben Morris and Phillip Rogaway.
\newblock {Sometimes-Recurse} shuffle: Almost-random permutations in
  logarithmic expected time.
\newblock In {\em Advances in cryptology---{EUROCRYPT} 2014}, volume 8441 of
  {\em Lecture Notes in Comput. Sci.}, pages 311--326. Springer, Heidelberg,
  2014.

\bibitem[NR99]{NaorReingold}
Moni Naor and Omer Reingold.
\newblock On the construction of pseudo-random permutations: {Luby-Rackoff}
  revisited.
\newblock {\em Journal of Cryptology}, 12:29--66, 1999.
\newblock Preliminary version in: \textit{Proc. STOC 97}.

\bibitem[PRS02]{PatelRamzanSundaram}
Sarvar Patel, Zulfikar Ramzan, and Ganapathy~S. Sundaram.
\newblock Luby-rackoff ciphers: Why {XOR} is not so exclusive.
\newblock In {\em Selected Areas in Cryptography, 9th Annual International
  Workshop, SAC 2002, St. John's, Newfoundland, Canada, August 15-16, 2002.
  Revised Papers}, pages 271--290, 2002.

\bibitem[RY13]{RY13}
Thomas Ristenpart and Scott Yilek.
\newblock {\em The {Mix-and-Cut} Shuffle: Small-Domain Encryption Secure
  Against N Queries}, pages 392--409.
\newblock Springer Berlin Heidelberg, Berlin, Heidelberg, 2013.

\bibitem[Vau98]{Vaudenay}
Serge Vaudenay.
\newblock {\em Provable security for block ciphers by decorrelation}, pages
  249--275.
\newblock Springer Berlin Heidelberg, Berlin, Heidelberg, 1998.

\bibitem[Zha12]{2Zhandry}
Mark Zhandry.
\newblock How to construct quantum random functions.
\newblock In {\em 53rd Annual Symposium on Foundations of Computer Science},
  pages 679--687, New Brunswick, NJ, USA, 2012. IEEE Computer Society Press.

\bibitem[Zha16]{Zhandry}
Mark Zhandry.
\newblock A note on quantum-secure {PRPs}.
\newblock {\em CoRR}, abs/1611.05564, 2016.

\end{thebibliography}

\clearpage
\appendix
\newpage
\section{Proof of probability of Game X}\label{ExplainX}
Recall the definition of $S^1_i, S^2_i, T^1_i$ and $T^2_i$ (see p.~\pageref{GamesXandR}.) We write $S_s$ and $T_t$ to denote the final transcripts. We drop the index $i$ if it is understood. We begin by defining \textbf{Game X'}.

\vspace*{0.015\textheight}
\hrule
\begin{footnotesize}
\begin{minipage}[t]{0.9\textwidth}
\vspace*{0.005\textheight}\textbf{GAME X':} Initially, let $S_0$ and $T_0$ be empty. Choose $k\in_R G$, then answer the $i+1$'st query as follows:

\textbf{$E$-oracle query with $m_{i+1}$:} \\
 \textbf{1.} If $P(m_{i+1}\cdot k)\in T^2_i$ return $P(m_{i+1}\cdot k)\cdot k$ \\
 \textbf{2.} Else choose $y_{i+1}\in_R G\setminus T^2_i$, define $P(m_{i+1}\cdot k) = y_{i+1}$, and return $y_{i+1}\cdot k$. \\

\textbf{$D$-oracle query with $c_{i+1}$:} \\
 \textbf{1.} If $P^{-1}(c_{i+1}\cdot k^{-1}) \in T^1_i$, return $P^{-1}(c_{i+1}\cdot k^{-1})\cdot k^{-1}$. \\
 \textbf{2.} Else choose $x_{i+1}\in_R G\setminus T^1_i$, define $P^{-1}(c_{i+1}\cdot k^{-1})=x_{i+1}$, and return $x_{i+1}\cdot k^{-1}$. \\

\textbf{$P$-oracle query with $x_{i+1}$:} \\
 \textbf{1.} If $P(x_{i+1})\in T^2_i$, return $P(x_{i+1})$. \\
 \textbf{2.} Else choose $y_{i+1}\in_R G\setminus T^2_i$, define $P(x_{i+1})=y_{i+1}$, and return $y_{i+1}$. \\

\textbf{$P^{-1}$-oracle query with $y_{i+1}$:} \\
 \textbf{1.} If $P^{-1}(y_{i+1})\in T^1_i$, return $P^{-1}(y_{i+1})$. \\
 \textbf{2.} Else choose $x_{i+1}\in_R G\setminus T^1_i$, define $P^{-1}(y_{i+1})=x_{i+1}$, and return $x_{i+1}$.
\vspace*{0.005\textheight}
\end{minipage}
\end{footnotesize}
\hrule
\vspace*{0.015\textheight}

Notice that the only difference between \textbf{Game X'} and the game defining $P_X$ is that the latter has defined all values for the oracles beforehand while the former "defines as it goes." Still, an adversary cannot tell the difference between playing the \textbf{Game X'} or the game defining $P_X$. Thus, $Pr_{X'}\left[ \mathcal{A}_{E,D}^{P,P^{-1}}=1 \right] = P_X$.

What we wish to show is that
\begin{align*}
Pr_X \left[ \mathcal{A}_{E,D}^{P,P^{-1}}=1\right] = Pr_{X'} \left[ \mathcal{A}_{E,D}^{P,P^{-1}}=1\right],
\end{align*}
i.e. that no adversary $\mathcal{A}$ may distinguish between playing \textbf{Game X} and playing \textbf{Game X'}, even negligibly. We will do this by showing that no adversary $\mathcal{A}$ may distinguish between the outputs given by the two games. As both games begin by choosing a uniformly random key $k$ and as we show that for this value the games are identical, we hereby assume such a key $k$ to be a fixed, but arbitrary, value for the remainder of this proof.

Considering the definitions of \textbf{Game X} and \textbf{Game X'}, we see that the two games define their $E/D$- and $P/P^{-1}$-oracles differently: the former defining both, while the latter defines only the $P/P^{-1}$-oracle and computes the $E/D$-oracle. We show that \textbf{Game X} also answers its $E/D$-oracle queries by referring to $P/P^{-1}$, although not directly.

Given the partial functions $E$ and $P$ in \textbf{Game X}, i.e. functions having been defined for all values up to and including the $i$'th query, define the partial function $\widehat{P}$ as the following.
\begin{align*}
\widehat{P}(x) \defeq 	\begin{cases}
							P(x) & \text{if } P(x) \text{ is defined,} \\
							E(x\cdot k^{-1}) \cdot k^{-1} & \text{if } E(x\cdot k^{-1}) \text{ is defined, and} \\
							\text{undefined} & \text{otherwise.}
						\end{cases}
\end{align*}
Using the above definition, defining a value for $E$ or $P$ implicitly defines a value for $\widehat{P}$. The first question is, whether or not $\widehat{P}$ is well-defined, i.e. whether there are clashes of values (that is, differences between values differing by other than $\cdot k$ (or $\cdot k^{-1}$)) for some $x$ for which both $P(x)$ and $E(x\cdot k^{-1})$ are defined.

\begin{lem}
Let $E$ and $P$ be partial functions arising in \textbf{Game X}, then the partial function $\widehat{P}$ is well-defined.
\end{lem}

\begin{proof}
Proof by induction on the number of "Define" steps in \textbf{Game X} (i.e. steps $E-3, D-3, P-3,$ and $P^{-1}-3$) as these are the steps where $\widehat{P}$ becomes defined. The initial case of the induction proof is trivial as $S_0$ and $T_0$ are empty such that no values may clash. Suppose now that in step $E-3$ we define $E(m)=c$. The only possibility that $\widehat{P}$ becomes ill-defined will occur if the new $E(m)$ value clashes with a prior defined $P(m\cdot k)$ value: If $P(m \cdot k)$ was not defined, then no clashes can arise. If $P(m\cdot k)$ was defined, then by step $E-2$, the value is $E(m)\cdot k^{-1}$, such that there is no clash.

For $D-3$, the argument is similar as $E(m)$ will become defined as well. Although, for the case where $P(m\cdot k)$ is defined, step $D-2$ forces a new uniformly random value of $m$ to be chosen until no clash occurs.

Analogously, for $P$ and $P^{-1}$, no clashes will arise, hence, $\widehat{P}$ must be well-defined.
\end{proof}

We may also consider $\widehat{P}$ in \textbf{Game X'}, in the sense that when we define a value for $P$ in the game, we implicitly define a value for $\widehat{P}$ where $\widehat{P}(x)=P(x)$ as $E(x\cdot k^{-1})=P(x)$ in \textbf{Game X'}.

We wish now to show that the oracle query-answers of $E, D, P,$ and $P^{-1}$ in \textbf{Game X}, expressed in terms of $\widehat{P}$, correspond exactly to those in \textbf{Game X'}.

\textbf{Case 1: $E$-oracle query.} Beginning with \textbf{Game X}, we first note that \textbf{Game X} never defines $E(m)$ unless $m$ has been queried to the $E$-oracle, or alternately, the $D$-oracle has been queried with a $c$ such that $E(m)=c$. However, as $\mathcal{A}$ never repeats a query if it can guess the answer, i.e. never queries any part of an already defined $E/D$-oracle pair, we may assume that $E(m)$ is undefined when $m$ is queried. Therefore, we see that concurrently with $m$ being queried, we have that $\widehat{P}(m\cdot k)$ will be defined if and only if $P(m\cdot k)$ is defined, and $\widehat{P}(m\cdot k) = P(m\cdot k)$. Let us consider the two cases: when $\widehat{P}(m\cdot k)$ is defined and when it is undefined.
\begin{itemize}
\item[\textit{Case 1a}:] When $\widehat{P}(m\cdot k)$ is defined, then \textbf{Game X} returns $c = \widehat{P}(m\cdot k)\cdot k$. Setting $E(m)=c$ leaves $\widehat{P}$ unchanged, i.e. the value $\widehat{P}(m\cdot k)$ remains the same, unlike the next case.
\item[\textit{Case 1b}:] \item[\textit{Case 1b}:] \begin{sloppypar} When $\widehat{P}(m \cdot k)$ is undefined, then \textbf{Game X} repeatedly chooses ${c\in_R G \setminus S^2}$ uniformly until $P^{-1}(c\cdot k^{-1})$ is undefined, i.e. the set $U=\lbrace c\in G | P^{-1}(c\cdot k^{-1})\not\in T^1\rbrace$. It follows that $y=c\cdot k^{-1}$ is uniformly distributed over $G \setminus \widehat{T}^2$.\footnote{$\widehat{T}^1$ and $\widehat{T}^2$ are the corresponding sets on the query pairs of $\widehat{P}$.} This can be seen by showing that $S^2 \cup U^\complement = \widehat{T}^2\cdot k$, where the only non-triviality in the argument follows from the definition of $\widehat{P}$. In this case, setting $E(m)=c$ also sets $\widehat{P}(m\cdot k)=y$, in contrast to the prior case as it is now defined. \end{sloppypar}
\end{itemize}
We now consider the same query on \textbf{Game X'}.
\begin{itemize}
\item[\textit{Case 1a'}:] When $\widehat{P}(m\cdot k)=P(m\cdot k)$ is defined, $c=P(m\cdot k)\cdot k$ is returned, and $\widehat{P}$ is unchanged.
\item[\textit{Case 1b'}:] When $\widehat{P}(m\cdot k)=P(m\cdot k)$ is undefined, we choose $y\in_R G\setminus T^2 = G\setminus \widehat{T}^2$, $\widehat{P}(m\cdot k)$ is set to $y$, and $c=y\cdot k$ is returned.
\end{itemize}
Thus, the behaviour of \textbf{Game X} and \textbf{Game X'} are identical on the $E$-oracle queries.

We will be briefer in our arguments for the following $3$ cases as the arguments are similar.

\textbf{Case 2: $D$-oracle query.} Here we again assume that no element of an $E/D$-oracle pair $(m,c)$, such that $E(m)=c$, has been queried before. Like in the above case, we see that, as $\widehat{P}(m \cdot k) = P(m\cdot k)$, we also have $\widehat{P}^{-1}(c \cdot k^{-1}) = P^{-1}(c\cdot k^{-1})$.
\begin{itemize}
\item[\textit{Case 2a $+$ 2a'}:] When $\widehat{P}^{-1}(c\cdot k^{-1})=P^{-1}(c\cdot k^{-1})$ is defined, then $m=P^{-1}(c\cdot k^{-1})\cdot k^{-1}$ is returned, leaving $\widehat{P}^{-1}(c \cdot k^{-1})$ unchanged in both games.
\item[\textit{Case 2b $+$ 2b'}:] If $\widehat{P}^{-1}(c\cdot k^{-1})=P^{-1}(c\cdot k^{-1})$ is undefined, then $x \in_R G \setminus \widehat{T}^1$ is chosen uniformly and $\widehat{P}^{-1}(c\cdot k^{-1}) = x$, in both cases.
\end{itemize}
Thus, the behaviour of \textbf{Game X} and \textbf{Game X'} are identical on the $D$-oracle queries.

\textbf{Case 3: $P$-oracle query.} Here we instead assume that no element of a $P/P^{-1}$-oracle pair $(x,y)$ such that $P(x)=y$, has been queried before.
\begin{itemize}
\item[\textit{Case 3a $+$ 3a'}:] Using the definition of the $E$- and $P$-oracles in \textbf{Game X} and the definition of $\widehat{P}$ we see that $P(x)$ is defined if and only if $E(x\cdot k^{-1})$ is defined, but then this also holds if and only if $\widehat{P}(x)$ is defined (by the assumption in the beginning of case $3$). Hence, if $\widehat{P}(x)$ is defined, then $y=E(x\cdot k^{-1})\cdot k^{-1}=\widehat{P}(x)$. Indeed, both games secure this value.
\item[\textit{Case 3b $+$ 3b'}:] If $\widehat{P}(x)$ is undefined, then $y \in_R G \setminus \widehat{T}^2$ is chosen uniformly and $\widehat{P}(x)$ is defined to be $y$, in both cases. 
\end{itemize}
Thus, the behaviour of \textbf{Game X} and \textbf{Game X'} are identical on the $P$-oracle queries.

\textbf{Case 4: $P^{-1}$-oracle query.} Again, we assume that no element of a $P/P^{-1}$-oracle pair $(x,y)$ such that $P(x)=y$, has been queried before.
\begin{itemize}
\item[\textit{Case 4a $+$ 4a'}:] Using the definition of \textbf{Game X} and the definition of $\widehat{P}$, as well as our case $4$ assumption, we see that $\widehat{P}^{-1}(y)$ is defined if and only if $D(y\cdot k)$ is defined. Hence, if $\widehat{P}^{-1}(y)$ is defined, then $x=D(y\cdot k) \cdot k = \widehat{P}^{-1}(y)$. Indeed, both games secure this value.
\item[\textit{Case 4b $+$ 4b'}:] If $\widehat{P}^{-1}(y)$ is undefined, then $x \in_R G \setminus \widehat{T}^1$ is chosen uniformly and $\widehat{P}^{-1}(y)$ is defined to be $x$, in both cases. 
\end{itemize}
Thus, the behaviour of \textbf{Game X} and \textbf{Game X'} are identical on the $P^{-1}$-oracle queries. Q.E.D.

\newpage
\section{Proof that the probability of Game R and Game R' match}\label{App:ReqRR}
Recall the definition of $S^1_i, S^2_i, T^1_i$ and $T^2_i$ (see p.~\pageref{GamesXandR}). We write $S_s$ and $T_t$ to denote the final transcripts. We also introduce the following definition.
\begin{defn}\label{overlapidentical}
We say that two $E/D$-pairs $( m_i,c_i )$ and $( m_j,c_j )$ \textbf{overlap} if $m_i=m_j$ or $c_i=c_j$. If $m_i=m_j$ and $c_i=c_j$, we say that the pairs are \textbf{identical}. Likewise for $P/P^{-1}$-pairs $( x_i,y_i )$ and $( x_j,y_j )$.
\end{defn}

If two pairs overlap, then by the definition of the $E/D$- and $P/P^{-1}$-oracles, they must be identical. Therefore, WLOG, we may assume that all queries to the oracles are non-overlapping. Let us now prove the lemma.

\begin{lem}
$Pr_R\left[ BAD \right] = Pr_{R'} \left[ BAD \right]$.
\end{lem}

\begin{proof}
We need to prove that \textbf{Game R} has its flag set to \textbf{bad} if and only if \textbf{Game R'} has its flag set to \textbf{bad}.
	
"$\Rightarrow$": We want to show that there exists $(m,c)\in S_s$ and $(x,y)\in T_t$ such that either $m\cdot k = x$ or $c\cdot k^{-1} = y$ (i.e. such that $k$ becomes bad). We have to consider the $8$ cases where the flag is set to bad. All of the cases use an analogous argument to the following: If $P(m \cdot k)$ is defined then $P(m\cdot k) = y = P(x)$ for some $(x,y)\in T_{t}$ such that, as overlapping pairs are identical, $m\cdot k = x$.

"$\Leftarrow$": We assume that there exists $(m,c)\in S_s$ and $(x,y)\in T_t$  such that $k$ becomes bad. i.e. such that either $m\cdot k = x$ or $c\cdot k^{-1} = y$. We need to check that in all four oracle queries, the flag in \textbf{Game R} is set to bad, which needs a consideration of 8 cases.

Assume that $m\cdot k = x$, then
\begin{align*}
E\text{-oracle on } m &: P(m\cdot k)=P(x) = y \in T_t^2, \\
D\text{-oracle on } c &: P(m\cdot k)=P(x) = y \in T_t^2, \\
P\text{-oracle on } x &: E(x\cdot k^{-1}) = E(m) = c \in S_s^2, \\
P^{-1}\text{-oracle on } y &: E(x\cdot k^{-1}) = E(m) = c \in S_s^2.
\end{align*}

Assume now that $c\cdot k^{-1} = y$, then
\begin{align*}
E\text{-oracle on } m &: P^{-1}(c\cdot k^{-1})=P^{-1}(y) = x \in T_t^1, \\
D\text{-oracle on } c &: P^{-1}(c\cdot k^{-1})=P^{-1}(y) = x \in T_t^1, \\
P\text{-oracle on } x &: D(y\cdot k) = D(c) = m \in S_s^1, \\
P^{-1}\text{-oracle on } y &: D(y\cdot k) = D(c) = m \in S_s^1.
\end{align*}
\end{proof}

\newpage
\section{Bijectivity of $h$}\label{Bijectivityofh}
We first note that $h^{-1}$ is well-defined by the same argument that $h$ is well-defined, but by considering the inverse shuffle.
If $b\in S_i$, then
\begin{align*}
h(h^{-1}(b)) = h(b) = b,
\end{align*}
as $b\in S_{i+1}$ per definition of $h^{-1}: S_{i+1} \rightarrow S_i$.
If $b\not\in S_i$, then
\begin{align*}
h(h^{-1}(b)) = h(k_{i+1}^{-1}\cdot b) = k_{i+1}\cdot k_{i+1}^{-1}\cdot b = b,
\end{align*}
because, as $k_{i+1}^{-1}\cdot b\in S_i$, either $k_{i+1}^{-1}\cdot b\in S_{i+1}$ or $k_{i+1}\cdot k_{i+1}^{-1}\cdot b = b \in S_{i+1}$, and we already know that $b\in S_{i+1}$ such that $k_{i+1}^{-1}\cdot b\not\in S_{i+1}$.

If $a\in S_{i+1}$, then
\begin{align*}
h^{-1}(h(a)) = h^{-1}(a) = a,
\end{align*}
per definition of $h: S_{i} \rightarrow S_{i+1}$.
If $a\not\in S_{i+1}$, then
\begin{align*}
h^{-1}(h(a)) = h^{-1}(k_{i+1}\cdot a) = k_{i+1}^{-1}\cdot k_{i+1}\cdot a = a,
\end{align*}
because, as $k_{i+1}\cdot a\in S_{i+1}$, either $k_{i+1}\cdot a\in S_{i}$ or $k_{i+1}^{-1}\cdot k_{i+1}\cdot a = a \in S_{i}$, but we already know that $a\in S_{i}$ such that $k_{i+1}\cdot a\not\in S_{i}$.

\end{document}